\documentclass[11pt]{article}

\usepackage[utf8]{inputenc}

\usepackage{amsmath,amsfonts,amsthm,amssymb,color}
\usepackage[usenames,dvipsnames,svgnames,table]{xcolor}
\definecolor{darkgreen}{rgb}{0.0,0,0.9}
\usepackage{mathtools}
\usepackage{authblk}
\usepackage{fullpage}
\usepackage{parskip}
\usepackage{comment}
\usepackage{tikz}
\usepackage{bbm}
\usepackage{dsfont}
\usepackage[sc]{mathpazo}
\usepackage[basic]{complexity}
\usepackage{algorithm2e}
\usepackage[colorlinks=true,
citecolor=OliveGreen,linkcolor=BrickRed,urlcolor=BrickRed,pdfstartview=FitH]{hyperref}
\usepackage[capitalize,nameinlink]{cleveref}
\usepackage{tcolorbox}
\usepackage[short]{optidef}

\newtcolorbox{wbox}
{
	colback  = white,
}

\SetKwInOut{Input}{Input}
\SetKwInOut{Output}{Output}
\SetKwFunction{Uncross}{\textsc{Uncross}}
\SetKwFunction{MergeUncross}{\textsc{MergeUncross}}
\SetKwFunction{PerfectMatching}{\textsc{PerfectMatching}}
\SetKwBlock{InParallel}{in parallel do}{end}
\SetKwFor{ParallelFor}{for}{in parallel do}{end}

\newcommand*{\suppress}[1]{}

\newcommand*{\cR}{\mathcal{R}}
\newcommand*{\cQ}{\mathcal{Q}}

\makeatletter
\def\thm@space@setup{%
	\thm@preskip= 10pt
	\thm@postskip=\thm@preskip 
}
\makeatother

\makeatletter
\renewcommand{\paragraph}{%
	\@startsection{paragraph}{4}%
	{\z@}{5pt}{-1em}%
	{\normalfont\normalsize\bfseries}%
}
\makeatother

\newtheorem{theorem}{Theorem}
\newtheorem{lemma}{Lemma}

\newtheorem{corollary}{Corollary}

\theoremstyle{definition}
\newtheorem{definition}{Definition}

\newtheorem{remark}{Remark}

\newtheorem{alg}{Algorithm}

\newtheorem{example}{Example}

\newenvironment{fminipage}%
{\begin{Sbox}\begin{minipage}}%
		{\end{minipage}\end{Sbox}\fbox{\TheSbox}}

\newcommand{\spread}{\mbox{\rm spread}}
\newcommand{\lleft}{\mbox{\rm left}}
\newcommand{\rright}{\mbox{\rm right}}

\newcommand{\Init}{\mbox{\textsc{Bin}}}
\newcommand{\Act}{\mbox{\textsc{Active}}}
\newcommand{\Full}{\mbox{\textsc{Fully-Repaired}}}
\newcommand{\Fro}{\mbox{\textsc{Frozen}}}

\newcommand{\Ssuba}{\mbox{\textsc{Min-Sub1}}}
\newcommand{\Ssubb}{\mbox{\textsc{Min-Sub2}}}
\newcommand{\Ssubc}{\mbox{\textsc{Min-Sub3}}}

\newcommand{\Xsuba}{\mbox{\textsc{Max-Sub1}}}
\newcommand{\Xsubb}{\mbox{\textsc{Max-Sub2}}}
\newcommand{\Xsubc}{\mbox{\textsc{Max-Sub3}}}

\title{Fair Core Imputations for the Assignment Game: \\ 
New Solution Concepts and Efficient Algorithms}

\author[1]{Vijay V.~Vazirani\footnote{Supported in part by NSF grant CCF-2230414.}}

\affil[1]{University of California, Irvine}

\date{}

\begin{document}
	\maketitle

	\begin{abstract}
	
The assignment game is a classical model for profit-sharing and a cornerstone of cooperative game theory. While an imputation in its {\bf core} guarantees fairness among coalitions, it provides no fairness guarantee at the level of individual agents: single agents or one-sided coalitions have zero standalone value and may receive arbitrarily small payoffs. Motivated by the growing focus on individual-level fairness, we ask: {\bf Can one select a core imputation that is also fair to individuals}?

We introduce {\bf three individual-fairness-driven solution concepts}, each promoting equity in a different way. The {\bf leximin} and {\bf leximax} core imputations extend max–min and min–max fairness to uplift the least advantaged and constrain the most advantaged agents, respectively. The {\bf min-spread} core imputation minimizes the gap between the largest and smallest positive payoffs, promoting equitable profit distribution.

For all three solution concepts, we develop {\bf combinatorial, strongly polynomial algorithms}. The leximin and leximax algorithms are based on a {\bf novel adaptation of the primal–dual paradigm}, while the min-spread algorithm combines partial executions of the first two. We expect our work to revive innovation on the potent primal-dual paradigm as well as promote further work on the algorithmic study of fairness and stability.

	\end{abstract}

\bigskip
\bigskip
\bigskip
\bigskip
\bigskip
\bigskip
\bigskip
\bigskip
\bigskip
\bigskip
\bigskip
\bigskip
\bigskip
\bigskip
\bigskip

\pagebreak

\clearpage
\pagenumbering{arabic}

\section{Introduction}
\label{sec.intro}

{\bf Fairness and equitable division} have long been central themes in economics, from early social choice theory \cite{Arrow2012social} to modern challenges in Internet resource allocation \cite{Moulin-Internet}. The study of fairness has come to a basic truth: {\em no single definition can universally capture fairness}, e.g., \cite{Fairness} lists 21 definitions! Within this vast literature, we are concerned with fairness at the level of {\em coalitions} and {\em individual fairness}. For the former, the {\em core} is the gold standard solution concept: it ensures that no coalition has an incentive to secede by guaranteeing each coalition at least as much profit as they can achieve on their own. Hence a core imputation guarantees stability by allocating a fair profit to each of exponentially many coalitions—a stringent requirement indeed. 

Recent work across game theory, networking and resource allocation has renewed attention on \emph{individual-level fairness} as a fundamental consideration, e.g., the wide spread use of max-min and min-max fairness \cite{max-min1, max-min2, max-min3, max-min4}, as well as  their extensions to \emph{leximin} and {\em leximax} allocations in modern applications such as fair rent division \cite{Procaccia-Rent} and citizens' assemblies \cite{Nature-2021}.

Our work bridges these two traditions in the context of the assignment game, which occupies a central place within cooperative game theory; the classic work of Shapley and Shubik \cite{Shapley1971assignment} gave a complete characterization of its core, see Section \ref{sec.SS-assn-game}. In this game, a coalition consisting of a single agent or a group from one side of the bipartition has zero standalone value. Consequently, under an arbitrary core imputation, individual agents are allowed to receive arbitrary profit shares, leading to inequity. To address this deficiency, we seek a core imputation that satisfies two key criteria:

\begin{enumerate}
\item The allocation of profit shares is {\bf individually fair}.
\item It is computable in {\bf polynomial—or preferably, strongly polynomial—time} via a {\bf combinatorial algorithm}\footnote{That is, the algorithm does not invoke an LP-solver, e.g., see \cite{Sch-book, LP.book}.}
\end{enumerate}


We propose {\bf three solution concepts}. All three deal with profit shares of {\em essential agents}, i.e., agents which are matched in every maximum weight matching, see Section \ref{sec.Complementarity}; these are the only agents which get positive profits in core imputations.  

\begin{enumerate}
	\item the {\bf leximin core imputation}, which extends max-min fairness by maximizing the smallest individual share and, subject to that, the next smallest, and so on;
	\item the {\bf leximax core imputation}, which symmetrically extends min-max fairness by minimizing the largest share, then the second largest, and so on;
	\item the {\bf min-spread core imputation}, which minimizes the difference between the largest and smallest  profit shares of essential agents.
\end{enumerate}

 It is easy to see that none of these solution concepts dominates the others, e.g., see the example given in Section \ref{sec.properties}. Our solution concepts seek to reconcile the tension between coalitions demanding profit which is commensurate with their inherent worth and the societal norm of ``equality'' as a desirable outcome. Each promotes individual fairness in a distinct way: leximin uplifts the least advantaged, leximax reins in the most advantaged and min-spread makes individual profit shares as equal as possible.  In Section \ref{sec.properties} we show that whereas leximin and leximax core imputations are unique, min-spread is not. In Section \ref{sec.comparison} we show that the prominent solution concept of {\bf nucleolus does not promote individual fairness} in any of these three ways, even for an infinite family of games, Theorem \ref{thm.infinite}.  

We design {\bf combinatorial, strongly polynomial algorithms}, having a running time of $O(m n^3)$, for all three proposed solution concepts. For the leximin and leximax core imputations, our algorithms are based on a novel adaptation of the classical primal-dual paradigm; see Section~\ref{sec.framework}. The algorithm for the min-spread solution is arguably even more innovative. It is made possible by the combinatorial nature of the previous two algorithms and is obtained by carefully orchestrating their partial executions, see Section~\ref{sec.min-spread}.

The assignment and stable matching games share a key structural property: in both, the set of solutions forms a {\bf distributive lattice}, with the top and bottom elements favoring opposite sides of the bipartition (see details in Section~\ref{sec.related}). These extreme solutions are efficiently computable in both settings. However, while the stable matching lattice is {\bf finite}—enabling the design of algorithms for finding more equitable outcomes—the lattice of core imputations in the assignment game is {\bf infinite}. This fundamental difference precludes the use of Birkhoff’s Representation Theorem and necessitates a {\bf fundamentally new approach} to identifying individually fair core imputations.



We anticipate that the novel way of running the primal-dual method proposed here will not only be exploited for additional applications, but it will also revive innovation on this highly potent method\footnote{This  is the reason we preferred submitting this paper to a theory, rather than an AGT, conference.}. Additionally, we expect our work to promote further work on the algorithmic study of fairness and stability, as was done in \cite{Leximin-max}.



\subsection{How was the Primal-Dual Framework Adapted to our Problem?}
\label{sec.framework}
 
 The primal-dual method has three distinguishing features:
 
 \begin{enumerate}
 	\item  It provides a {\bf flexible, high level framework} whose details can be tailored to the specific problem at hand. 
 	\item It achieves {\bf global optimality or near-optimality} through a sequence of {\bf local improvements}.
 	\item It leads to efficient, {\bf combinatorial} algorithms. 
 \end{enumerate} 
 
 The method achieved major successes, first in {\bf combinatorial optimization} \cite{LP.book, Sch-book} and then in {\bf approximation algorithms} \cite{Va.ApproximationAlgs}. These settings involved finding an {\bf  integral} optimal (or near-optimal) solution to the primal LP, together with an optimal (or near-optimal) solution to the dual LP. Complementary slackness conditions --- or an appropriate relaxation of these conditions --- helped transform the global task into a sequence of local tasks. Following these successes, researchers adapted it to solve {\em rational convex programs}, i.e.,  non-linear convex programs that always admit rational solutions, defined in \cite{Va.rational}. These include \cite{DPSV, JV.Equitable, Va.rational, duan2016improved}.  Indeed, today it is appropriate to call this method the {\bf primal-dual paradigm}.

 Despite its versatility, it was not initially clear that this method would be suitable for our setting. In our case, computing an optimal solution to the primal and dual LPs --- namely, a maximum-weight matching and an associated core imputation --- is straightforward. However, the core challenge lies in transforming this arbitrary core imputation into the leximin or leximax core imputation.
 
 One might consider locally improving the imputation while maintaining the optimality of both primal and dual solutions. However, because the primal-dual pair satisfies all complementary slackness conditions, how will the global problem be broken into local steps? To address this, we introduce a novel coordinating mechanism: a  {\bf clock} that governs the algorithm's progress and synchronizes all updates. For each time $t$, we can define a vertex set $W_t \subseteq (U \cup V)$ such that in the induced subgraph $H_t$ of $G$, the profit shares form the {\bf unique leximin core imputation} for $H_t$\footnote{The proof of this fact crucially depends on the correct resolution of some subtle algorithmic choices, see Sections \ref{sec.legitimate} and \ref{sec.right}.}. We prove that in polynomial time $W_t = (U \cup V)$, hence terminating in the required imputation. The proof of this fact crucially depends on the {\em correct resolution of some subtle algorithmic choices}, see Sections \ref{sec.legitimate} and \ref{sec.right}.



While prior primal-dual algorithms have employed the notion of time, our use of it is distinct in both purpose and prominence. For example, in \cite{JV.Facility}, dual variables evolve over time, but this evolution is used merely to enforce relaxed complementary slackness conditions, which remain the primary engine for local-to-global transformation. Similarly, in \cite{JV.Equitable}, time progression aids in ensuring cross-monotonicity in cost-sharing methods, rather than in structuring the optimization process itself. In contrast, in our approach, time is the central coordinating device that transforms the global optimization problem into a sequence of structured local updates.

The broad outline given above is far from sufficient for designing the entire algorithm. The latter requires a precise understanding of the way core imputations allocate the total profit to agents and teams; in particular,  we need to answer the following questions:

\begin{enumerate}
	\item Do core imputations spread the profit more-or-less evenly or do they restrict profit to certain well-chosen agents? If the latter, what characterizes these ``chosen'' agents? 
	\item  An edge $(i, j)$ in the underlying graph is called a {\em team}. By definition, under any core imputation, the sum of profits of $i$ and $j$ is at least the weight of this edge. For which teams is the sum strictly larger? 
	\item How do core imputations behave in the presence of degeneracy? An assignment game is said to be {\em non-degenerate} if it admits a unique maximum weight matching. 
	\end{enumerate}

Our answers, using complementary slackness conditions and strict complementarity, are provided in Section \ref{sec.Complementarity}.

\section{Related Works}
\label{sec.related}

{\bf Comparison with the stable matching game:}
A useful perspective is provided by the stable matching game of Gale and Shapley \cite{GaleS}, which is a non-transferable utility (NTU) game. The core of this game is the set of all stable matchings --- each ensures that no coalition formed by one agent from each side of the bipartition has an incentive to secede. It is well known that this set forms a {\bf  finite distributive lattice}. Its top and bottom elements maximally favor one side and disfavor the other side of the bipartition, and the Deferred Acceptance Algorithm \cite{GaleS} computes precisely these two extremes.
 
However, in many applications, one desires stable matchings that are more equitable to both sides. This has motivated considerable research on finding such matchings efficiently, leading to efficient algorithms as well as intractability results, see \cite{Irving-egalitarian, Teo1998geometry, Cheng-Median, Cheng-Center, Two-sided-MMbook}. These efficient algorithms crucially use {\bf Birkhoff's Representation Theorem} for finite distributive lattices \cite{Birkhoff} and the notion of rotations \cite{Rotation-Irving}. 
 
 Shapley and Shubik \cite{Shapley1971assignment} showed that the set of its core imputations is the entire polytope of optimal solutions to the dual of the LP-relaxation of the maximum weight matching problem in the underlying bipartite graph. If this set is not a singleton, it is a dense set consisting of uncountably many elements. Furthermore, this set forms an {\bf infinite distributive lattice}. The top and bottom elements of this lattice form two extreme imputations, each maximally favoring one side and disfavoring the other side of the bipartition. So far efficient algorithms were known for computing only these two (least fair) imputations as outlined below\footnote{Their convex combination is also a core imputation; however, it is not clear what fairness properties it possesses.}. The current paper offers a recourse. 
 
 Define $u^h$ and $u^l$ as follows: For $i \in U$, let $u_i^h$ and $u_i^l$ denote the highest and lowest profits that $i$ accrues among all imputations in the core. Similarly, define $v^h$ and $v^l$. Shapley and Shubik  showed that the set $S$ of core imputations of the assignment game forms a distributive lattice; its two extreme imputations are $(u^h, v^l)$ and $(u^l, v^h)$; these can be efficiently obtained: Find a maximum weight (fractional) matching in $G$ and let its picked edges be $M \subseteq E$. Then $(u^h, v^l)$, is obtained by solving LP (\ref{eq.extreme2}) and and $(u^l, v^h)$ is obtained by solving an analogous LP.

 	\begin{maxi}
		{} {\sum_{i \in U}  {u_{i}}} 
			{\label{eq.extreme2}}
		{}
		\addConstraint{ u_i + v_j}{ = w_{ij} \quad }{\forall (i, j) \in M}
		\addConstraint{ u_i + v_j}{ \geq w_{ij} \quad }{\forall (i, j) \in E}
		\addConstraint{u_{i}}{\geq 0}{\forall i \in U}
		\addConstraint{v_{j}}{\geq 0}{\forall j \in V}
	\end{maxi}

{\bf Comparison with Dutta and Ray's egalitarian solution:}
In a prominent work, Dutta and Ray \cite{Dutta-Ray} showed that a convex game\footnote{The characteristic function of such a game is supermodular, see Definition \ref{def.cooperative-game}.} admits a unique core imputation which Lorenz dominates all other core imputations; they call it the {\bf egalitarian solution}. From our standpoint, \cite{Dutta-Ray} has three shortcomings: it does not apply to several key natural games, including the assignment, MST and max-flow games, since they are not convex; it is not efficiently computable for any non-trivial, natural game; and the Lorenz order is a partial, and not a total, order. Our work rectifies all  three shortcomings.

\medskip


{\bf Other related results:} 
The importance of ``fair'' profit-sharing has been a consideration since ancient times, e.g., see the delightful and thought-provoking paper by Aumann and Maschler \cite{Aumann1985game}, which discusses problems from the 2000-year-old Babylonian Talmud on how to divide the estate of a dead man among his creditors. 

Below we state results giving LP-based core imputations for natural games; these are also candidates for the study of leximin and leximax core imputations. Deng et al. \cite{Deng1999algorithms} gave LP-based characterization of the cores of several combinatorial optimization games, including maximum flow in unit capacity networks both directed and undirected, maximum number of edge-disjoint $s$-$t$ paths, maximum number of vertex-disjoint $s$-$t$ paths, maximum number of disjoint arborescences rooted at a vertex $r$, and general graph matching game in which the weight of optimal fractional and integral matchings are equal (such a game has a non-empty core). 
 
Regarding the facility location game, first, Kolen \cite{Kolen-facility} showed that for the unconstrained facility location problem, each optimal solution to the dual of the classical LP-relaxation is a core imputation if and only if this relaxation has no integrality gap. Later, Goemans and Skutella \cite{Goemans-Skutella} showed a similar result for any kind of constrained facility location game. They also proved that in general, for facility locations games, deciding whether the core is non-empty and whether a given allocation is in the core is NP-complete. 

Samet and Zemel \cite{Samet-Zemel} study games which are generated by linear programming optimization problems; these are called LP-games. For such games, It is well known that the set of optimal dual solutions is contained in the core and \cite{Samet-Zemel} gives sufficient conditions under which equality holds. These games do not ask for integral solutions and are therefore different in character from the ones studied in this paper. 

Granot and Huberman \cite{Granot1981minimum, Granot-2-1984core} showed that the core of the minimum cost spanning tree game is non-empty and gave an algorithm for finding an imputation in it. Koh and Sanita \cite{Laura-Sanita} answer the question of efficiently determining if a spanning tree game is submodular; the core of such games is always non-empty. Nagamochi et al. \cite{Nagamochi1997complexity} characterize non-emptyness of core for the minimum base game in a matroid; the minimum spanning tree game is a special case of this game. \cite{Transportation-core} gives a characterization of the Owen core \cite{Owen.core}, i.e., the optimal dual solutions, of the transportation game, which is the same as the $b$-matching game 

To deal with games having an empty core, the following two notions have been given. The first is that of {\bf least core}, defined by Mascher et al. \cite{Leastcore-Maschler1979geometric}. If the core is empty, there will necessarily be sets $S \subseteq V$ such that $v(S) < p(S)$ for any imputation $v$. The least core maximizes the minimum of $p(S) - v(S)$ over all sets $S \subseteq V$, subject to $v(\emptyset) = 0$ and $v(V) = p(V)$. 

A more well known notion is that of {\bf nucleolus} \cite{Schmeidler1969nucleolus} which is contained in the least core, see Definition \ref{def.nucleolus}. For any imputation, $p$, define $p(S) - c(S)$ to be the {\bf surplus} of coalition $\emptyset \subset S \subset V$. The idea behind the nucleolus is to spread the surplus as evenly as possible across all coalitions by picking the {\bf leximin imputation over the surplus of coalitions}. The nucleolus always exists for any cooperative game, even if its core is empty and is unique. Furthermore, if the core of the game is non-empty, then the nucleolus lies in the core. Solymosi and Raghavan \cite{Solymosi-Nucleolus} gave a strongly polynomial algorithm for computing the nucleolus of the assignment game, see also Section \ref{sec.properties}.

\section{Definitions and Preliminary Facts}
\label{sec.prelim}

After giving some basic definitions, we will state Shapley and Shubik's results for the assignment game as well as insights obtained 
via complementarity. For an extensive coverage of cooperative game theory, see the book by Moulin \cite{Moulin2014cooperative}. The assignment game is widely applicable, e.g., \cite{Shapley1971assignment} presented it in the context of the housing market game in which the two sides of the bipartition are buyers and sellers of houses.

\begin{definition}
In a {\bf transferable utility (TU) market game}, utilities of the agents are stated in monetary terms and side payments are allowed. 
\end{definition}

\begin{definition}
	\label{def.cooperative-game}
	A {\bf cooperative game} consists of a pair $(N, c)$ where $N$ is a set of $n$ agents and $c$ is the {\bf characteristic function}; $c: 2^N \rightarrow \cQ_+$, where for $S \subseteq N, \ c(S)$ is the {\bf worth}\footnote{Consistent with the usual presentation of efficient algorithms and realistic models of computation, we will assume that all numbers given as inputs are rational, i.e., finite precision. This will reflect on most of our definitions as well.} that the coalition $S$ can generate by itself. $N$ is also called the {\bf grand coalition}. We will denote the worth of the game, i.e., $c(N)$, by $W_{\max}$. 
\end{definition}

\begin{definition}
	\label{def.imputation}	
	An {\bf imputation}
	is a function $p: N \rightarrow \cQ_+$ that gives a way of dividing the worth of the game, $c(N)$, among the agents, i.e., it satisfies $\sum_{i \in N} {p(i)} = c(N) = W_{\max}$. $p(i)$ is called the {\bf profit} of agent $i$. For a coalition $S \subseteq N$, define $p(S) = \sum_{i \in S} {p(i)}$, i.e., the total profit given to its agents. 
\end{definition}

\begin{definition}
	\label{def.core}
	An imputation $p$ is said to be in the {\bf core of the game} $(N, c)$ if for any coalition $S \subseteq N$, the total profit allocated to agents in $S$ is at least as large as the worth that they can generate by themselves, i.e., $ {p(S)} \geq c(S)$.
\end{definition}

\subsection{The Results of Shapley and Shubik for the Assignment Game}
\label{sec.SS-assn-game}


The {\bf assignment game} consists of a bipartite graph $G = (U, V, E)$ and a weight function for the edges $w: E \rightarrow \cQ_+$, where we assume that {\bf each edge has positive weight}; otherwise the edge can be removed from $G$. The agents of the game are the vertices, $U \cup V$ and the {\bf total worth} of the game is the weight of a maximum weight matching in $G$; the latter needs to be distributed among the vertices. 

For this game, a coalition $(U' \cup V')$ consists of a subset of the agents, also called vertices,  with $U' \subseteq U$ and $V' \subseteq V$. The {\bf worth} of a coalition $(U' \cup V')$ is defined to be the weight of a maximum weight matching in the graph $G$ restricted to vertices in $(U' \cup V')$ and is denoted by $c(U' \cup V')$. The {\bf characteristic function} of the game is defined to be $c: 2^{U \cup V} \rightarrow \cQ_+$. 

An {\bf imputation}	consists of two functions $u: {U} \rightarrow \cQ_+$ and $v: {V} \rightarrow \cQ_+$ such that $\sum_{i \in U} {u_i} + \sum_{j \in V} {v_j} = c(U \cup V)$. Imputation $(u, v)$ is said to be in the {\bf core of the assignment game} if for any coalition $(U' \cup V')$, the total profit allocated to agents in the coalition is at least as large as the worth that they can generate by themselves, i.e., $\sum_{i \in U'} {u_i} +  \sum_{j \in V'} {v_j} \geq c(S)$. 

Linear program (\ref{eq.core-primal-bipartite}) gives the LP-relaxation of the problem of finding a maximum weight matching in $G$. In this program, variable $x_{ij}$ indicates the extent to which edge $(i, j)$ is picked in the solution. 
	\begin{maxi}
		{} {\sum_{(i, j) \in E}  {w_{ij} x_{ij}}}
			{\label{eq.core-primal-bipartite}}
		{}
		\addConstraint{\sum_{j: (i, j) \in E} {x_{ij}}}{\leq 1 \quad}{\forall i \in U}
		\addConstraint{\sum_{i: (i, j) \in E} {x_{ij}}}{\leq 1 }{\forall j \in V}
		\addConstraint{x_{ij}}{\geq 0}{\forall (i, j) \in E}
	\end{maxi}
Taking $u_i$ and $v_j$ to be the dual variables for the first and second constraints of (\ref{eq.core-primal-bipartite}), we obtain the dual LP: 
 	\begin{mini}
		{} {\sum_{i \in U}  {u_{i}} + \sum_{j \in V} {v_j}} 
			{\label{eq.core-dual-bipartite}}
		{}
		\addConstraint{ u_i + v_j}{ \geq w_{ij} \quad }{\forall (i, j) \in E}
		\addConstraint{u_{i}}{\geq 0}{\forall i \in U}
		\addConstraint{v_{j}}{\geq 0}{\forall j \in V}
	\end{mini}

The constraint matrix of LP (\ref{eq.core-primal-bipartite}) is totally unimodular, i.e., its every square sub-matrix has determinant of 0, $+1$ or $-1$. This fact leads to the next theorem, see \cite{LP.book}.

\begin{theorem}
	\label{thm.int-assn-LP}
The polytope defined by the constraints of LP (\ref{eq.core-primal-bipartite}) is integral; its vertices are matchings in the underlying graph. 
\end{theorem}

The proof of the next theorem hinges on the key fact stated in Theorem \ref{thm.int-assn-LP}.

\begin{theorem}
	\label{thm.SS}
	(Shapley and Shubik \cite{Shapley1971assignment})
	The dual LP, (\ref{eq.core-dual-bipartite}) completely characterizes the core of the assignment game.
\end{theorem}

By this theorem, the core is a convex polyhedron. Shapley and Shubik further showed that the set $S$ of core imputations of the assignment game forms a lattice under the following partial order on $S$: Given two core imputations $(u, v)$ and $(u', v')$, say that $(u, v) \leq (u', v')$ if for each $i \in U$, $u_i \geq u_i'$; if so, it must be that for each $j \in V, v_j \leq v_j'$. This partial order supports meet and join operations: given two core imputations $(u, v)$ and $(u', v')$, for each $i \in U$, the meet chooses the bigger of $u_i$ and $u_i'$ and the join chooses the smaller. It is easy to see that the lattice is distributive. Define $u^h$ and $u^l$ as follows: For $i \in U$, let $u_i^h$ and $u_i^l$ denote the highest and lowest profits that $i$ accrues among all imputations in the core. Similarly, define $v^h$ and $v^l$.

\begin{theorem}
	\label{thm.extreme}
		(Shapley and Shubik \cite{Shapley1971assignment})
Under the partial order defined above, the core of the assignment game forms a distributive lattice; its two extreme imputations are $(u^h, v^l)$ and $(u^l, v^h)$.
\end{theorem}

The two extreme imputations can be efficiently obtained as follows. First find a maximum weight matching in $G$; a fractional one found by solving LP (\ref{eq.core-primal-bipartite}) will suffice. Let its picked edges be $M \subseteq E$. The $U$-optimal core imputation, $(u^h, v^l)$, is obtained by solving LP (\ref{eq.extreme}):

 	\begin{maxi}
		{} {\sum_{i \in U}  {u_{i}}} 
			{\label{eq.extreme}}
		{}
		\addConstraint{ u_i + v_j}{ = w_{ij} \quad }{\forall (i, j) \in M}
		\addConstraint{ u_i + v_j}{ \geq w_{ij} \quad }{\forall (i, j) \in E}
		\addConstraint{u_{i}}{\geq 0}{\forall i \in U}
		\addConstraint{v_{j}}{\geq 0}{\forall j \in V}
	\end{maxi}

\section{Complementarity Applied to the Assignment Game}
\label{sec.Complementarity}

 In this section, we provide answers to the three questions which were raised in the Section \ref{sec.framework}.   The design of our primal and dual update steps critically uses these insights. 
 
  Observe that the worth of the assignment game is given by an optimal solution to the primal LP and its core imputations are given by optimal solutions to the dual LP. Hence the fundamental fact connecting these two solutions is complementary slackness conditions. Although this fact is well known, somehow its implications  were not explored before\footnote{The conference version of the results of this section appeared in \cite{Va.New-characterizations}.}.

\subsection{The first question: Allocations made to agents by core imputations}
\label{sec.vertices}

\begin{definition}
	\label{def.agent}
	A generic agent in $U \cup V$ will be denoted by $q$. We will say that $q$ is:
	\begin{enumerate}
		\item {\em essential} if $q$ is matched in every maximum weight matching in $G$.
		\item {\em viable} if there is a maximum weight matching $M$ such that $q$ is matched in $M$ and another, $M'$ such that $q$ is not matched in $M'$. 	
		\item {\em subpar} if for every maximum weight matching $M$ in $G$, $q$ is not matched in $M$. 	
		\end{enumerate}
\end{definition}

\begin{definition}
\label{def.agent-paid}
	Let $y$ be an imputation in the core. We will say that $q$ {\em gets paid in $y$} if $y_q > 0$ and {\em does not get paid} otherwise. Furthermore, $q$ is {\em paid sometimes} if there is at least one imputation in the core under which $q$ gets paid, and it is {\em never paid} if it is not paid under every imputation. 
\end{definition}

\begin{theorem}
	\label{thm.vertices}
	 For every agent $q \in (U \cup V)$: 
		\[ q \ \mbox{is paid sometimes}  \ \iff \ q \ \mbox{is essential} \]  
\end{theorem}
	
\begin{proof}
The proof follows by applying complementary slackness conditions and strict complementarity to the primal LP (\ref{eq.core-primal-bipartite}) and dual LP (\ref{eq.core-dual-bipartite}); see \cite{Sch-book} for formal statements of these facts. By Theorem \ref{thm.SS}, talking about imputations in the core of the assignment game is equivalent to talking about optimal solutions to the dual LP.

 Let $x$ and $y$ be optimal solutions to LP (\ref{eq.core-primal-bipartite}) and LP (\ref{eq.core-dual-bipartite}), respectively. By the Complementary Slackness Theorem, for each $q \in (U \cup V): \ y_q \cdot (x(\delta(q)) - 1) = 0$. 

$(\Rightarrow)$  Suppose $q$ is paid sometimes. Then, there is an optimal solution to the dual LP, say $y$, such that $y_q > 0$. By the Complementary Slackness Theorem, for any optimal solution, $x$, to LP (\ref{eq.core-primal-bipartite}), $x(\delta(q)) = 1$, i.e., $q$ is matched in $x$. Varying $x$ over all optimal primal solutions, we get that $q$ is always matched. In particular, $q$ is matched in all optimal assignments, i.e., integral optimal primal solutions, and is therefore essential. This proves the forward direction.

$(\Leftarrow)$ Strict complementarity implies that corresponding to each agent $q$, there is a pair of optimal primal and dual solutions, say $x$ and $y$, such that either $y_q = 0$ or $x(\delta(q)) = 1$ but not both. Assume that $q$ is essential, i.e., it is matched in every integral optimal primal solution. 
 
 We will use Theorem \ref{thm.int-assn-LP}, which implies that every fractional optimal primal solution to LP (\ref{eq.core-primal-bipartite}) is a convex combination of integral optimal primal solutions. Therefore $q$ is fully matched in every optimal solution, $x$, to LP (\ref{eq.core-primal-bipartite}), i.e., $x(\delta(q)) = 1$, so there must be an optimal dual solution $y$ such that $y_q > 0$. Hence $q$ is paid sometimes, proving the reverse direction. 
\end{proof}

Theorem \ref{thm.vertices} is equivalent to the following. For every agent $q \in (U \cup V)$: 
		\[ q \ \mbox{is never paid} \ \iff \ q \ \mbox{is not essential} \]  
		
Thus core imputations pay only essential agents and each of them is paid in some core imputation.  Since we have assumed that the weight of each edge is positive, so is the worth of the game, and all of it goes to essential agents. Hence we get:

\begin{corollary}
	\label{cor.vertices}
	In the assignment game, the set of essential agents is non-empty and in every core imputation, the entire worth of the game is distributed among essential agents; moreover, each of them is paid in some core imputation. 
\end{corollary} 

By Corollary \ref{cor.vertices}, the leximin and leximax core imputations are only concerned with the profits of essential agents; the profits of the rest of the agents are zero. 

Corollary \ref{cor.vertices} raises the following question: Can't a non-essential agent, say $q$, team up with another agent, say $p$, and secede, by promising $p$ almost all of the resulting profit? The answer is ``No'', because the dual (\ref{eq.core-dual-bipartite}) has the constraint $y_q + y_p \geq w_{qp}$. Therefore, if $y_q = 0$, $y_p \geq w_{q p}$, i.e., $p$ will not gain by seceding together with $q$.

\subsection{The second question: Allocations made to teams by core imputations}
\label{sec.edges}

\begin{definition}
	\label{def.team}
	By a {\em team} we mean an edge in $G$; a generic one will be denoted as $e = (u, v)$. We will say that $e$ is:
	\begin{enumerate}
		\item {\em essential} if $e$ is matched in every maximum weight matching in $G$.
		\item {\em viable} if there is a maximum weight matching $M$ such that $e \in M$, and another, $M'$ such that $e \notin M'$. 
		\item {\em subpar} if for every maximum weight matching $M$ in $G$, $e \notin M$. 
	\end{enumerate}
	\end{definition}
	
\begin{definition}
\label{def.team-paid}
	 Let $y$ be an imputation in the core of the game. We will say that $e$ is {\em fairly paid in $y$} if $y_u + y_v = w_e$ and it is {\em overpaid} if $y_u + y_v > w_e$\footnote{Observe that by the first constraint of the dual LP (\ref{eq.core-dual-bipartite}), these are the only possibilities.}. Finally, we will say that $e$ is {\em always paid fairly} if it is fairly paid in every imputation in the core.
\end{definition}

\begin{theorem}
	\label{thm.edges}
	 For every team $e \in E$: 
		\[ e \ \mbox{is always paid fairly} \ \iff \ e \ \mbox{is viable or essential} \]
\end{theorem}
	
\begin{proof}
The proof is similar to that of Theorem \ref{thm.vertices}. Let $x$ and $y$ be optimal solutions to LP (\ref{eq.core-primal-bipartite}) and LP (\ref{eq.core-dual-bipartite}), respectively. By the Complementary Slackness Theorem, for each $e = (u, v) \in E: \ \ x_e \cdot (y_u + y_v - w_e) = 0$.

$(\Leftarrow)$ To prove the reverse direction, suppose $e$ is viable or essential. Then there is an optimal solution to the primal, say $x$, under which it is matched. Therefore,  $x_e > 0$. Let $y$ be an arbitrary optimal dual solution. Then, by the Complementary Slackness Theorem, $y_u + y_v = w_e$, i.e., $e$ is fairly paid in $y$. Varying $y$ over all optimal dual solutions, we get that $e$ is always paid fairly. 

$(\Rightarrow)$ To prove the forward direction, we will use strict complementarity. It implies that corresponding to each team $e$, there is a pair of optimal primal and dual solutions $x$ and $y$ such that either $x_e = 0$ or $y_u + y_v = w_e$ but not both. 

Assume that team $e$ is always fairly paid, i.e., under every optimal dual solution $y$, $y_u + y_v = w_e$. By strict complementarity, there must be an optimal primal solution $x$ for which $x_e > 0$. Theorem \ref{thm.int-assn-LP} implies that $x$ is a convex combination of optimal assignments. Therefore, there must be an optimal assignment in which $e$ is matched. Therefore $e$ is viable or essential and the forward direction also holds. 
\end{proof}

\begin{corollary}
	\label{cor.endpoint-essential}
	Let $e = (u, v)$ be a subpar team. Then at least one of $u$ and $v$ is essential. 
\end{corollary}

\begin{proof}
In every maximum weight matching, at least one of $u$ and $v$ must be matched, since otherwise $e$ should get matched. By Theorem \ref{thm.edges}, there is a core imputation under which $e$ is overpaid. Therefore, under this imputation, at least one of $u$ and $v$ is paid and by Theorem \ref{thm.vertices}, that vertex must be essential. 
\end{proof}

Negating both sides of the implication proved in Theorem \ref{thm.edges}, we get the following  implication. For every team $e \in E$: 
		\[ e \ \mbox{is subpar} \ \iff \ e \ \mbox{is sometimes overpaid} \]

Clearly, this statement is equivalent to the statement proved in Theorem \ref{thm.edges} and hence contains no new information. However, it provides a new viewpoint. These two equivalent  statements yield the following assertion, which at first sight seems incongruous with what we desire from the notion of the core and the just manner in which it allocates profits:

\begin{center}
{\em Whereas viable and essential teams are always paid fairly, subpar teams are sometimes overpaid.}
\end{center}

How can the core favor subpar teams over viable and essential teams? An explanation is provided in the Introduction, namely a subpar team $(i, j)$ gets overpaid because $i$ and $j$ create worth by playing in competent teams with other agents. Finally, we observe that contrary to Corollary \ref{cor.vertices}, which says that the set of essential agents is non-empty, it is easy to construct examples in which the set of essential teams may be empty. 

By Theorem \ref{thm.vertices}, different essential agents are paid in different core imputations and by Theorem \ref{thm.edges}, different subpar teams are overpaid in different core imputations. This raises the following question: is there a core imputation that simultaneously satisfies all these conditions? Lemma  \ref{lem.simultaneous} gives a positive answer. 

\begin{lemma}
		\label{lem.simultaneous}
	For the assignment game, there is a core imputation satisfying:
\begin{enumerate}
	\item a agent $q \in U \cup V$ gets paid if and only if $q$ is essential.
	\item a team $e \in E$ gets overpaid if and only if $e$ is subpar. 
\end{enumerate}
\end{lemma}

\begin{proof}
By Theorem \ref{thm.vertices}, for each essential agent $q$, there is a core imputation under which $q$ gets paid and by Theorem \ref{thm.edges}, for each subpar team $e$, there is a core imputation under which $e$ gets overpaid. Consider a convex combination of all these imputations; it must give positive weight to each of these imputations. Clearly, this is a core imputation. 

We observe that none of the core imputations pay non-essential agents or overpay non-subpar teams. Consequently, the imputation constructed above satisfies the conditions of the theorem. 
\end{proof}

The first part of Lemma \ref{lem.simultaneous} provides useful information about the leximin core imputation, see Corollary \ref{cor.positive}; its second part is given for completeness. Let $I$ denote the assignment game $G = (U, V, E)$ with weight function $w: E \rightarrow \cQ_+$ for edges. Let $C(I)$ denote its set of core imputations of $I$ and let $V' \subseteq V$ be its set of essential players. 

\begin{definition}
	\label{def.max-min}
Imputation $p \in C(I)$ is said to be a {\em max-min fair core imputation} if it satisfies:
$$ p \in \arg \max_{q \in C(I)} \left\{ \min_{v \in V'} \{q(v)\} \right\} .$$
\end{definition}

\begin{corollary}
	\label{cor.positive}
In the leximin core imputation, the profits of all essential agents are positive. 
\end{corollary}
	
\begin{proof}
	By the first part of Lemma \ref{lem.simultaneous}, the profit of the minimum profit essential agent in a max-min fair core imputation is positive. Since this is also the smallest component of the leximin core imputation, the corollary follows.
\end{proof}

\subsection{The third question: Degeneracy}
\label{sec.degeneracy}

The answer to the third question is useful for pedagogical purposes, in case the reader wishes to understand the algorithm for the case of a non-degenerate game. In this case, the underlying structure is simpler: there are no viable teams or agents. This follows from Theorems \ref{thm.vertices} and \ref{thm.edges}. How do viable teams and agents behave with respect to core imputations if the given game is degenerate? The answer is provided below.  
 
\begin{corollary}
	\label{cor.degen}
	In the presence of degeneracy, imputations in the core of an assignment game treat:
	\begin{itemize}
			\item  viable agents in the same way as subpar agents, namely they are never paid.  
		\item viable teams in the same way as essential teams, namely they are always fairly paid. 
	\end{itemize}
\end{corollary}

\begin{definition}
	\label{def.degen}
An assignment game is said to be {\bf non-degenerate} if it has a unique maximum weight matching. If so, by Theorem \ref{thm.vertices} and \ref{thm.edges}, all vertices and edges are either essential or subpar. 
\end{definition}

\section{Solution Concepts: Properties and Examples} 
\label{sec.properties}

In this section, we will formally define our three solution concepts and establish basic properties; in Example \ref{ex.different}, all three yield distinct imputations. This example was made possible by insights gained through the combinatorial nature of our algorithms. We will also compare our proposed solutions with two prominent alternatives --- the {\bf nucleolus} and the core imputation that maximizes {\bf Nash social welfare} --- and show why the latter options fall short on our criteria.  
  
 A leximin allocation maximizes the smallest component and subject to that, it maximizes the second smallest, and so on; leximax is analogous. Let $U_e \subseteq U$ and $V_e \subseteq V$ be the sets of essential players of the assignment game, $G = (U, V, E), \ w: E \rightarrow \cQ_+$. As stated in Theorem \ref{thm.vertices}, every core imputation allocates the entire profit to players in $U_e \cup V_e$ only.

\begin{definition}
	\label{def.leximin}
For each core imputation of the assignment game $G = (U, V, E), \ w: E \rightarrow \cQ_+$, sort the allocations made to players $U_e \cup V_e$ in increasing order. Then the imputation yielding the lexicographically largest list will be called a {\bf leximin core imputation}. Next, sort the allocations in decreasing order and the imputation yielding the lexicographically largest list will be called a {\bf leximax core imputation}. 
\end{definition}

\begin{definition}
	\label{def.spread}
Define the {\bf spread} of core imputation $p$ to be the difference between the largest and smallest profit given to essential players, i.e., 
$$  \spread(p)  :=  (p_{\max} - p_{\min}), $$ 
$$ \mbox{where} \ \ p_{max} = \max_{a \in (U_e \cup V_e)} {p(a)} \ \ \mbox{and} \ \ p_{min} = \min_{a \in (U_e \cup V_e)}  {p(a)} .$$ 
An imputation that minimizes the spread among all core imputations is called a {\bf min-spread core imputation}.  
\end{definition}


\begin{example}
	\label{ex.different}
In the graph\footnote{In the examples, essential and viable edges are represent by solid lines and subpar edges by dotted lines.} shown in Figure \ref{fig.3-SC}, the unique leximin and leximax core imputations for vertices $(u_1, u_2, u_3, u_4, v_1, v_2, v_3, v_4)$ are $(100, 36, 28, 90, 36, 40, 40, 28)$ and $(70, 6, 3, 65, 66, 70, 65, 53)$, respectively; the corresponding sorted orders are $28, 28, 36, 36, 40, 40, 90, 100$ and \\ $70, 70, 66, 65, 65, 53, 6, 3$. A min-spread imputation is $(92, 28, 28, 90, 44, 48, 40, 28)$, having spread 64. The last imputation was obtained via the min-spread algorithm, given in Section \ref{sec.min-spread}, as follows: Assume that Phase 1 of this algorithm ends with the leximin imputation. Phase 2 terminates in Case 1 with $\Omega = 92$ and outputs the min-spread imputation given above. 
\end{example}


\begin{figure}[h]
\begin{center}
\includegraphics[width=3in]{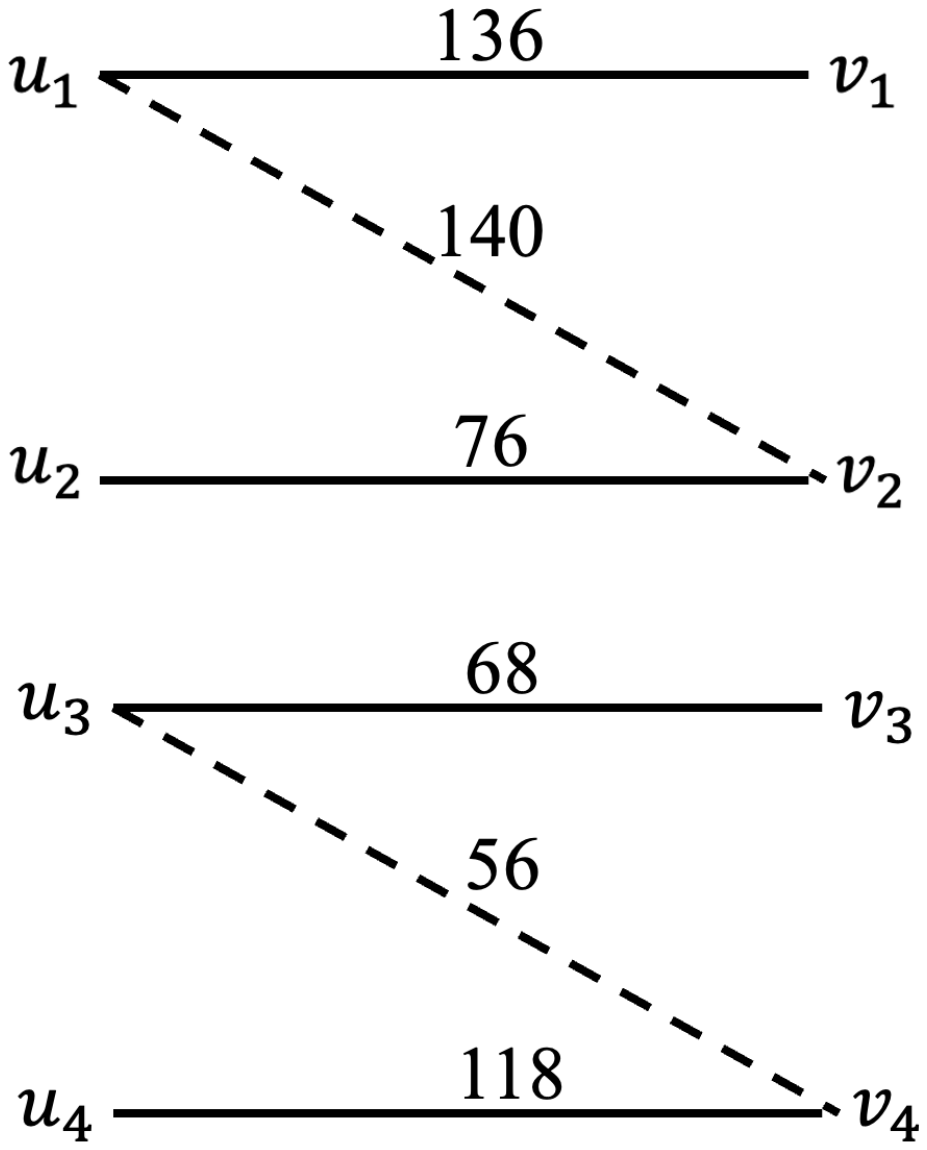}
\caption{The graph for Example \ref{ex.different}.}  
\label{fig.3-SC}
\end{center}
\end{figure}


Moulin \cite{Moulin-axioms} showed that the set of leximin core imputations for a convex game is a singleton, i.e., these games have a unique leximin imputation. In Lemma \ref{lem.leximin-unique} we show this fact for any cooperative game. The proof for leximax core imputations is analogous. In contrast, as shown in Example \ref{ex.not-unique}, a min-spread core imputation is not unique. 

We first need an elementary fact. Let $L$ be a list of numbers, $a_1 \leq a_2 \leq \ldots \leq a_k$ such that at least two of them are distinct. Consider a {\bf cycle} formed by edges $(a_i, a_{i+1})$, for $1 \leq i \leq k$ and $(a_k, a_1)$. For each of these $k$ edges, compute the average of its endpoints and let $L'$ denote this list of numbers. 

\begin{lemma}
	\label{lem.cycle}
$L'$ is lexicographically larger than $L$. 
\end{lemma}

\begin{proof}
	Let $l$ be the smallest index such that $a_l < a_{l+1}$. Clearly, the average on the first $l-1$ edges will be $a_1$ and on the remaining edges, it will be strictly bigger. Therefore, after sorting $L'$, its first $l-1$ entries will be the same as that of $L$ and its $l^{th}$ entry will be strictly bigger. The lemma follows. 
\end{proof}

Let $(N, c)$ be a cooperative game having a non-empty core, where $N$ is a set of $n$ agents and $c$ is the characteristic function. First observe that if $p_1: N \rightarrow \cQ_+$ and $p_2: N \rightarrow \cQ_+$ are core imputations for this game then their convex combination is also a core imputation, since the profit given to any coalition $S \subseteq N$ is at least $c(S)$. Therefore, the core of the game $(N, c)$ is a convex set.

\begin{lemma}
	\label{lem.leximin-unique}
A game $(N, c)$ having a non-empty core admits a unique leximin imputation.
\end{lemma}

\begin{proof}
The proof follows via a contradiction. Suppose there are two leximin imputations $p_1$ and $p_2$.  Let $p_3: N \rightarrow \cQ_+$ denote the average of $p_1$ and $p_2$. Since the core is closed under convex combinations, $p_3$ is also a core imputation. We will show that it is lexicographically larger than $p_1$ and $p_2$, hence giving a contradiction.  

Since $p_1$ and $p_2$ are both leximin imputations, the multi-set of profit shares under $p_1$ and $p_2$ must the same, say given by sorted list $L$. Let $L'$ denote the sorted list of profit shares under $p_3$. Let $S_1$ and $S_2$ denote the agents of $N$ sorted by increasing profit shares according to $p_1$ and $p_2$, respectively. The $i^{th}$ agent of $S_1$ is denoted by $S_1(i)$; similarly for $S_2$. Since the list of profit shares of the two imputations is the same, $p_1(S_1(i)) = p_2(S_2(i))$. 

If $S_1(i) = S_2(i)$, we will say that the $i^{th}$ entry of $L$ is {\bf settled}, since it will carry over unchanged to $L'$. Next assume that $S_1(i) \neq S_2(i)$. Then $i$ will participate in a {\bf cycle}, i.e., there are indices $1 \leq l_1 < l_2 < \ldots l_k \leq n$ such that $i$ is one of these indices and
\begin{enumerate}
	\item for $1 \leq j \leq k, \ l_j$ is not settled, i.e., $S_1(l_j) \neq S_2(l_{j})$, 
	\item  for $ \  1 \leq j < k, \ S_1(l_j) = S_2(l_{j+1})$. Furthermore $\ S_1(l_k) = S_2(l_1) $. 
\end{enumerate}

Clearly, the $k$ entries $l_1 < l_2 < \ldots l_k$ in list $L'$ will be given by averaging operation of Lemma \ref{lem.cycle} and therefore they will be lexicographically larger than the corresponding $k$ entries in $L$. Furthermore, this is true for each cycle formed by $S_1$ and $S_2$.The settled entries of $L$ carry over to $L'$. Hence $L'$ is lexicographically larger than $L$.
\end{proof}

\begin{corollary}
	\label{cor.unique}
	The assignment game has a unique leximin core imputation.
\end{corollary}


\begin{example}
	\label{ex.not-unique}
To demonstrate that the min-spread core imputation is not unique, we augment the graph in Figure~\ref{fig.3-SC} by adding an essential edge $(u_5, v_5)$ with weight 120. The imputation $(92, 28, 28, 90, 44, 48, 40, 28)$ for $(u_1, u_2, u_3, u_4, v_1, v_2, v_3, v_4)$ can be enhanced with a range of possibilities for $u_5$ and $v_5$ ---  all the way from $(28, 92)$ to $(92, 28)$ --- while maintaining the sum of profit shares at 120. Each possibility is a min-spread imputation. 
	
\end{example}

\subsection{A Comparison with Two Well-Known Solution Concepts}
\label{sec.comparison}

\begin{definition}
	\label{def.nucleolus}
	Let $(N, c)$ be a cooperative game with non-empty core. 
	For an imputation $p: {V} \rightarrow \cR_+$, let $\theta(p)$ be the vector obtained by sorting the $2^{|V|} - 2$ values $p(S) - c(S)$ for each $\emptyset \subset S \subset V$ in non-decreasing order. Then the unique imputation, $p$, that lexicographically maximizes $\theta(p)$ is called the {\bf nucleolus}.\end{definition}

Although the definition of nucleolus calls for spreading the surplus evenly over all coalitions, for the assignment game, a key simplification occurs: once the surplus of singletons and doubletons (i.e., teams) is fixed, the surplus of all coalitions is determined. For this reason, the nucleolus of the assignment game can be computed in strongly polynomial time via a combinatorial algorithm \cite{Solymosi-Nucleolus}, thereby meeting our second criterion\footnote{In other natural games, the nucleolus attempts to spread the profit over all exponentially many coalitions and is known to be NP-hard, e.g., MST game \cite{MST_nucleolus_np-hard}, max-flow game \cite{Flow_nucleolus_np-hard}, and $b$-matching game when an edge can be matched at most once \cite{b-matching-nucleolus-NP}.}. 

However, the nucleolus does not satisfy our first criterion for two key reasons. First, there are quadratically more doubletons than singletons. Second, since the worth of a singleton is zero, its entire profit directly counts as surplus, while the surplus of a doubleton is inherently more constrained. As a result, the nucleolus disproportionately focuses on balancing surpluses among doubletons, often at the expense of fairness to individual agents. Even in a simple instance (Example~\ref{ex.min-max}), the nucleolus fails to meet the individual fairness conditions captured by any of our three proposed solution concepts. In this example, it is not even max-min or min-max fair, hence it neither benefits low-earning agents nor reduces the payoffs of high-earning agents. Moreover, it produces a relatively large spread of 60, resulting in significant disparity in individual payoffs. Theorem \ref{thm.infinite} shows this for an infinite family of games.

\begin{definition}
	\label{def.NSW}
A core imputation which maximizes the product of profits of essential agents is called a {\bf maximum Nash social welfare (NSW) core imputation} for the assignment game. 
\end{definition}

The product of $k$ numbers having a fixed total is maximized by making them as equal as possible; hence the Nash social welfare is attempting to make the individual profit shares as equal as possible, e.g., in Example \ref{ex.min-max}, the spread of the NSW imputation is 40, same as that of leximin and leximax imputations. 

A core imputation of Definition \ref{def.NSW} is captured by a solution to convex program (\ref{eq.Nash}). Notice that the objective function is based on profits of essential agents only. As shown in Example \ref{ex.min-max}, this is not a rational convex program, as defined in \cite{Va.rational}, and therefore it does not admit a combinatorial polynomial time algorithm. Hence it does not satisfy our second criterion.

 	\begin{maxi}
		{} {\sum_{i \in U_e}  {\log(u_{i})} + \sum_{j \in V_e}  {\log(v_{i})}} 
			{\label{eq.Nash}}
		{}
		\addConstraint{\sum_{i \in U_e}  {(u_{i})} + \sum_{j \in V_e}  {(v_{i})}}{ = W_{\max}}
		\addConstraint{ u_i + v_j}{ \geq w_{ij} \quad }{\forall (i, j) \in E}
		\addConstraint{u_{i}}{\geq 0}{\forall i \in U_e}
		\addConstraint{v_{j}}{\geq 0}{\forall j \in V_e}
		\addConstraint{u_{i}}{= 0}{\forall i \in (U - U_e)}
		\addConstraint{v_{j}}{= 0}{\forall j \in (V - V_e)}
	\end{maxi}

\begin{example}
	\label{ex.min-max}
In the graph shown in Figure \ref{fig.min-max}, the leximin and leximax core imputations for vertices $(u_1, u_2, v_1, v_2)$ are $(40, 30, 30, 70)$ and $(55, 45, 15, 55)$, respectively; the corresponding sorted orders are $30, 30, 40, 70$ and $55, 55, 45, 15$. Both are min-spread imputations, having a spread of 40. The nucleolus assigns to $(u_1, u_2, v_1, v_2)$ profits of $(50, 20, 20, 80)$ as can be checked by considering the 10 singletons and doubletons. It has a high spread of 60, thereby creating disparity in profits. Additionally, it is neither max-min nor min-max fair, thereby it neither helps the low profit agents nor penalizes high earning agents. 

For computing the core imputation that maximizes the Nash social welfare, assume that the profit share of vertex $v_1$ is $x$. Then it is easy to see that the profit shares of $(u_1, u_2, v_2)$ will be $(70-x, \ 60-x, \ x+40)$. The product of these four values will be maximized for an irrational $x$. The values of $(u_1, u_2, v_1, v_2)$ were numerically computed (using Wolfram Alpha) and are approximately $(43.521, 33.521, 26.479, 66.479)$. 
\end{example}


\begin{figure}[h]
\begin{center}
\includegraphics[width=2.2in]{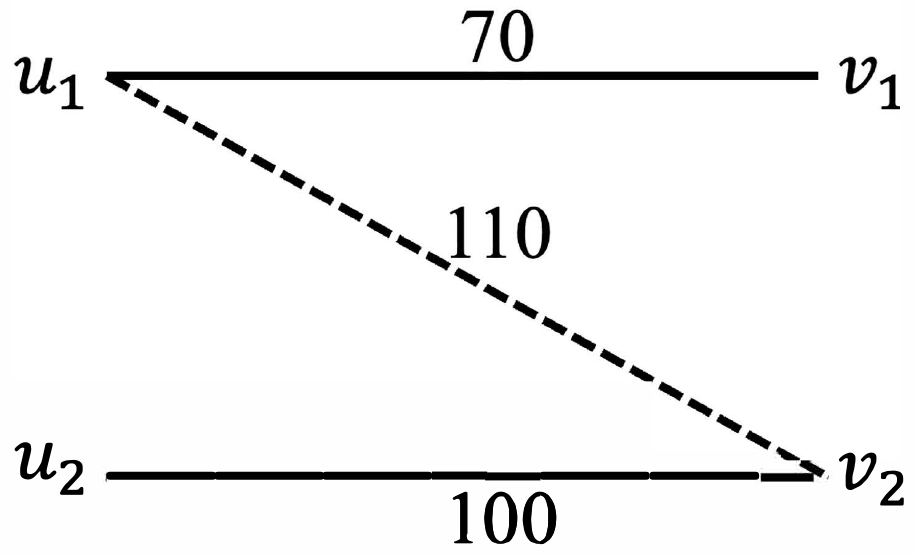}
\caption{The graph for Example \ref{ex.min-max}}.  
\label{fig.min-max}
\end{center}
\end{figure}


\begin{theorem}
		\label{thm.infinite}
There is an infinite family of assignment games in each of which the imputation given by the nucleolus does not satisfy any of our three solution concepts for individual level fairness.
\end{theorem}

\begin{proof}
	Extend the graph of Example \ref{ex.min-max} into an infinite family of games, $G_1, G_2, \ldots$, by adding in game $G_k$ a matching of $k$ edges of weight 60 each. 

In the game $G_k$, the nucleolus assigns to $(u_1, u_2, v_1, v_2)$ profits of $(50, 20, 20, 80)$ and it assigns profits of 30 to each endpoint of the added matching, leading to a spread of 60. This imputation is not max-min or min-max fair for the same reason as given in Example \ref{ex.min-max}. To show that it is not min-spread, consider the imputation that assigns to $(u_1, u_2, v_1, v_2)$ profit shares of $(40, 30, 30, 70)$, and 30 to each endpoint of the added matching, leading to a spread of 40. 
\end{proof}


\section{High Level Algorithmic Ideas}
\label{sec.high}

The ideas underlying Algorithm \ref{alg.min} are presented in this section and Section \ref{sec.remaining-leximin}. Our algorithm starts by computing an arbitrary core imputation, say $(u. v)$. We will say that an {\bf edge is tight} if it is fairly paid; as shown in Section \ref{sec.Complementarity}, essential and viable edges are fairly paid. The following invariant ensures that the algorithm always works with a core imputation.   

{\bf Invariant 1:}  {\bf Throughout the run of the algorithm, the dual satisfies:} 

\begin{enumerate}
	\item  Every essential and viable edge is tight.
	\item  Only essential vertices have positive duals. 
\end{enumerate}

For a proof of the next lemma, see Section \ref{sec.Classify}. 

\begin{lemma}
	\label{lem.classify}
The classification of edges and vertices into essential, viable and subpar can be accomplished in $O(m)$ calls to a weighted bipartite matching algorithm. 
\end{lemma}

$T_0$ denotes the {\bf set of tight edges}, and is initialized with all essential and viable edges. Let $H_0 = (U, V, T_0)$. We will partition the connected components of $H_0$ into the two sets defined below. 

\begin{definition}
	\label{def.unique-imp}
A connected component, $C$, of $H_0$, is called a {\bf unique imputation component} if the profit shares of vertices of $C$ remain unchanged over all possible core imputations.
\end{definition}

\begin{definition}
	\label{def.fundamental}
A connected component, $C$, of $H_0$ which is not a unique imputation component is called a {\bf fundamental component}.
\end{definition}

We will first explain these ideas in the simpler setting that {\bf the given game is non-degenerate} and therefore has no viable vertices and edges, see Definition \ref{def.degen}, and only in the end, in Section \ref{sec.char}, will we deal with the general case. Since $G$ is non-degenerate, unique imputation components are unmatched vertices and fundamental components are essential edges. {\bf Fundamental components are the building blocks which our algorithm works with}; the profit shares of their  vertices get modified in the dual updates. This update gets blocked when a subpar edge goes exactly tight. Much of the algorithm is concerned with appropriately dealing with the latter event.

\begin{definition}
\label{def.min-both}
For any connected component $C$ of the tight subgraph $H$, $U(C)$ and $V(C)$ will denote the vertices of $C$ in $U$ and $V$, respectively. Furthermore,	$\lleft(C)$ will denote $U(C) \cap U_e$ and $\rright(C)$ will denote $V(C) \cap V_e$. 
\end{definition}


\begin{definition}
\label{def.min-both}
	At any point in the algorithm, for any connected component $C$ of $H$, $\min(C)$ will denote the {\bf smallest profit share of an essential vertex in $C$}. We will say that {\bf $C$ has min on both sides} if there is a vertex in $\lleft(C)$ as well as one in $\rright(C)$ whose profit is $\min(C)$; if so, $C$ is said to be {\bf fully repaired}. Otherwise, we will say that {\bf $C$ has min in $\lleft(C)$} ($\rright(C)$) if there is a vertex in $\lleft(C)$ ($\rright(C)$) whose profit is $\min(C)$.
\end{definition}

The nature of the leximin order dictates that in order to leximin-improve the current core imputation, we {\bf must first try to increase the profit share of essential vertices having minimum profit}; we call the  process {\bf repairing the fundamental component}. Let $C$ be a fundamental component which has $C$ has min in $\lleft(C)$, see Definition \ref{def.min-both}. Since the essential edges of $C$ need to always be tight, the only  way of repairing $C$ is to rotate $C$ clockwise, see Definition \ref{def.rotate}. If $C$ has min in $\rright(C)$, we will need to rotate $C$ anti-clockwise.

\begin{definition}
	\label{def.rotate}
	Let $C$ be a connected component in the tight subgraph. By {\bf rotating $C$ clockwise} we mean continuously increasing the profit of $\lleft(C)$ and decreasing the profit of $\rright(C)$ at unit rate, and by {\bf rotating $C$ anti-clockwise} we mean continuously decreasing the profit of $\lleft(C)$ and increasing the profit of $\rright(C)$ at unit rate. A component $C$ will be repaired until either $C$ has min on both sides or a subpar edge $e$ incident at $C$ goes exactly tight. 
\end{definition}

{\bf Two Key Ideas:} \\
{\bf 1).} A central role is played by variable $\Omega$ which helps {\bf synchronize} all events, as mentioned in Section \ref{sec.framework}. $\Omega$ is initialized to the minimum profit share of an essential vertex under the starting core imputation, $(u, v)$. Thereafter, $\Omega$ is {\bf raised at unit rate} until the algorithm terminates. Therefore, in effect $\Omega$ defines a {\bf notion of time} and we will call it {\bf the clock}. $\Omega$ provides a way of synchronizing the simultaneous repair of several components. 

{\bf 2).} Can tight subpar edges be added to $T$ when they go exactly tight? If not, why not and which tight subpar edges need to be added to $T$ and when? Our algorithm adds a tight subpar edge to the tight subgraph only if it is {\bf legitimate}, see Definition \ref{def.legitimate}. This rule is critical for ensuring correctness of the algorithm, see Section \ref{sec.legitimate}.

\begin{definition}
	\label{def.legitimate}
Let $(i, j)$ be a tight subpar edge, with $i \in C$ and $j \in C'$, where connected component $C \in \Act$. We will say that $(i, j)$ is {\bf legitimate} if either $C$ has min in $\lleft(C)$ and $i \in \rright(C)$ or $C$ has min in $\rright(C)$ and $i \in \lleft(C)$.
\end{definition}

At any point in the algorithm, $G$ is partitioned into four sets, $\Init, \Fro$, $\Full, \Act$. They are initialization as follows: all unique imputation components are moved into $\Fro$ and all fundamental components are moved into $\Init$, see the pseudocode in Algorithm \ref{alg.min}. The arrows in Figure \ref{fig.flowchart} indicate the flow of fundamental components between these four sets. The purpose of $\Act$ is to carry out repair of fundamental components. As $\Omega$ increases, at some point one of three possible events may happen: The first is that $\Omega$ may become equal to $\min(C)$, for some component $C \in \Init$. If so, $C$ gets moved from $\Init$ to $\Act$. The second and the third events deal with dual and primal updates and are described in detail in Sections \ref{sec.dual-update} and \ref{sec.primal-update}, respectively.


\begin{figure}[h]
\begin{center}
\includegraphics[height=2.2in]{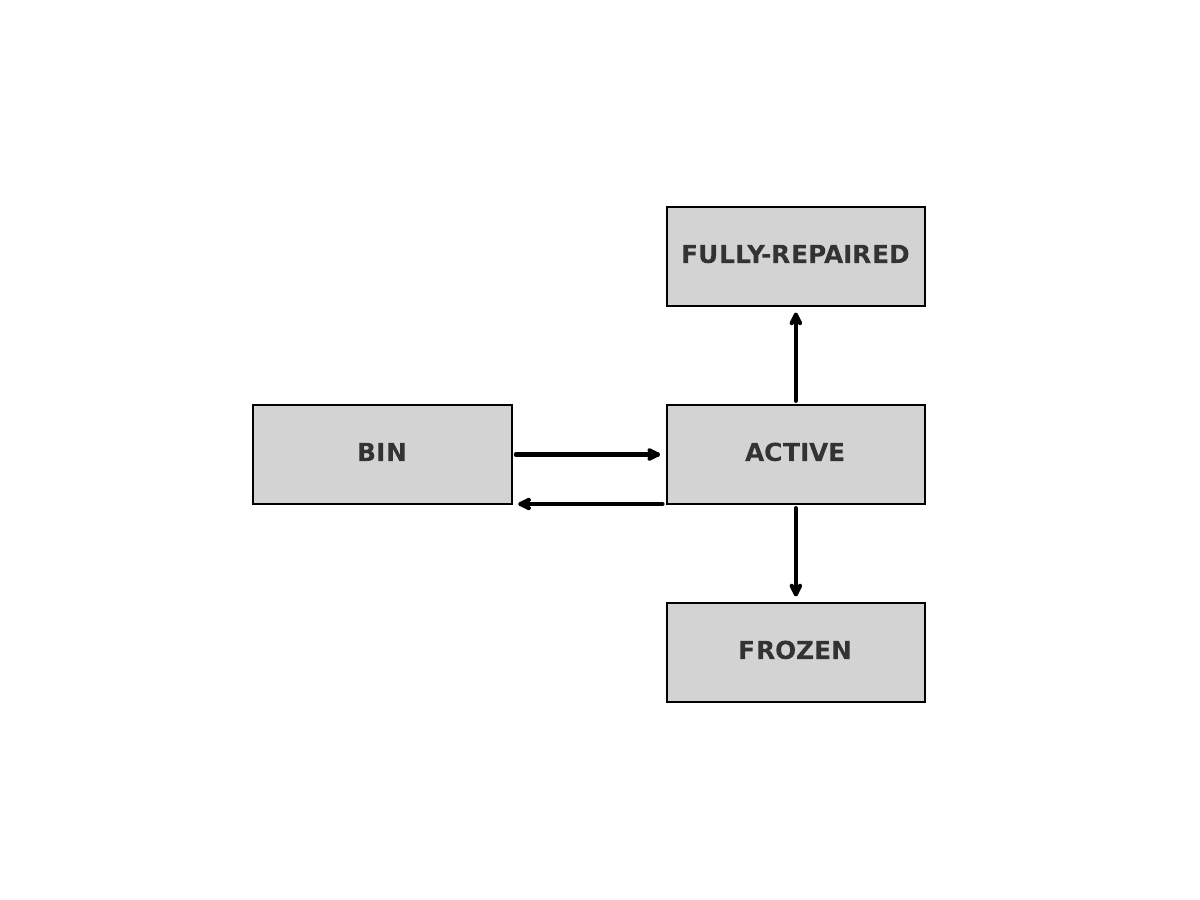}
\caption{Arrows indicate the flow of fundamental components in Algorithm \ref{alg.min}.} 
\label{fig.flowchart}
\end{center}
\end{figure}



\subsection{Dual Update: Repairing Valid Components} 
\label{sec.dual-update}

Each connected component in $\Act$ is called {\bf valid} and its form is described in Definition \ref{def.valid-comp}.

\begin{definition}
	\label{def.valid-comp}
A {\bf valid component}, $D$, consists of either a single fundamental component, $C_0$, or a set of fundamental components, $C_0, C_1, \ldots, C_k$, connected up via legitimate tight subpar edges. Because of the way $D$ is formed, it satisfies the following properties.
\begin{enumerate}
	\item Shrink each fundamental component, $C_i$, to a single vertex, $w_i$, and if a legitimate edge runs between $C_i$ and $C_j$, then add an edge between $w_i$ and $w_j$. Then the resulting graph should form a tree, say $\tau$, which is rooted at $w_0$, where $w_0$ is obtained by shrinking $C_0$; for a proof, see Lemma \ref{lem.tree-valid}. If $w_j$ is a child of $w_i$ in $\tau$, then we will say that {\bf component $C_j$ is a child of $C_i$}. The reflexive, transitive closure of the relation child will be called {\bf descendent}. Thus all fundamental components in $D$ are descendent of $C_0$. Let $w$ be a descendent of $w_0$ in $\tau$ and let the unique path from $w_0$ to $w$ in $\tau$ be $w_0, w_1, \ldots , w_l = w$. The  fundamental components $C_0, C_1, \ldots C_l$, together with the subpar edges between consecutive components, is called the {\bf path from $C_0$ to $C_l$}.  
	\item $\min(C_0) = \min_{i = 0}^k {\min(C_i)}$, and $C_0$ does not have min on both sides.
	\item First assume $C_0$ has min in $\lleft(C_0)$. Each subpar edge of $D$ was legitimate at the time it was added to $T$. Therefore, by Definition \ref{def.legitimate}, if component $C_j$ is a child of $C_i$ then the edge connecting them must run between $\rright(C_i)$ and $\lleft(C_j)$. 
		\item Next assume $C_0$ has min in $\rright(C_0)$. In this case, if component $C_j$ is a child of $C_i$ then the edge connecting them must run between $\lleft(C_i)$ and $\rright(C_j)$. 
	\item The composite component $D$ does not have min on both sides. Clearly, $\min(D) = \min(C_0)$. 
\end{enumerate}
\end{definition}

The repair of a valid component $D$ is similar to that of a fundamental component, namely if $D$ has min in $\lleft(D)$ then rotate $D$ clockwise and otherwise rotate $D$ anti-clockwise.  Assume that a valid component, $D$, has min on both sides. If so, should we declare $D$ fully repaired? It turns out that only an appropriate sub-component of $D$, $\Ssuba(D)$ should be declared fully repaired and moved into $\Full$, and the algorithm decomposes $(D - \Ssuba(D))$ into fundamental components and moves them back to $\Init$, see full details in  Section \ref{sec.right}. Identifying the right one is critical, as shown in Example \ref{ex.fully-repaired}. 

\begin{figure}

	\begin{wbox}
		\begin{alg}
		\label{alg.min}
		{\bf (Algorithm for Leximin Core Imputation)}\\

\begin{enumerate}
	\item {\bf Initialization:} 
	\begin{enumerate}
		\item $(u, v) \leftarrow$ an arbitrary core imputation. 		
		\item $T_0 \leftarrow $ essential and viable edges; $H_0 = (U, V, T_0)$. 
		\item $\Fro \leftarrow $ unique imputation components of $H_0$.
		\item $\Init \leftarrow $ fundamental components of $H_0$. 
		\item $\Omega \leftarrow$ minimum profit share of an essential vertex under $(u, v)$.  
		\item $T \leftarrow T_0$. 
	\end{enumerate}
			
\bigskip
		
	\item {\bf While} $(\Init \cup \Act) \neq \emptyset$  {\bf do}: \\
		   At unit rate, raise $\Omega$ and repair all components $C \in \Act$.  \\
           \hspace*{2mm}  {\bf If}: 
		\begin{enumerate}
		\item $\exists C \in \Init \ s.t. \ \min(C) = \Omega$ {\bf then} move $C$ to $\Act$
		\item $\exists C \in \Act$ is fully repaired {\bf then} move $\Ssuba(C)$ to $\Full$ \\
		and the fundamental components of $(C - \Ssuba(C))$ to $\Init$. 
		\item $\exists$ legitimate tight edge $(i, j)$. \\
		Assume its endpoints are in $C$ and $C'$, where $C \in \Act$.  
		
			 \begin{enumerate} 
			 \item  $T \leftarrow T \cup \{(i, j)\}$. 
			 \item Merge $C, C'$ and $(i, j)$ into one component, say $D$. 
			 \item {\bf If} $C' \in \Init$ {\bf then} move $D$ to $\Act$.   
			 \item {\bf If} $C' \in \Act$ {\bf then} move $\Ssubb(D)$ to $\Full$\\ and the fundamental components of $(D - \Ssubb(D))$ to $\Init$.
			 \item {\bf If} $C' \in \Fro$ {\bf then} move $\Ssubc(D)$ to $\Fro$ \\ and the fundamental components of $(D - \Ssubc(D))$ to $\Init$. 
			 \item {\bf If} $C' \in \Full$ {\bf then} move $\Ssubc(D)$ to $\Full$ \\ and the fundamental components of $(D - \Ssubc(D))$ to $\Init$.
				\end{enumerate} 
				
		\item Cleanup $T$: remove all tight subpar edges\\  \hspace*{1.88cm} connecting fundamental components in $\Init$.

			\end{enumerate} 
		\item {\bf Output:} Output the current imputation and HALT.  
	\end{enumerate} 
			\end{alg}
	\end{wbox}
\end{figure}

\bigskip

\section{Remaining Ideas of Leximin Algorithm}
\label{sec.remaining-leximin}

\subsection{Procedures for Classifying Vertices and Edges}
\label{sec.Classify}

In this section, we will give polynomial time procedures for partitioning vertices and edges according to the classification given in Definitions \ref{def.agent} and \ref{def.team}. This is required by the algorithms given in Sections \ref{sec.high} and \ref{sec.leximax}.

Let $W_{\max}$ be the worth of the given assignment game, $G = (U, V, E), \ w: E \rightarrow \cQ_+$. For each edge $e$, find the worth of the game with $e$ removed. If it is less than $W_{max}$ then $e$ is essential. Now pick an edge $e'$ which is not essential and let its weight be $w$. Remove the two endpoints of $e'$ and compute the worth of the remaining game. If it is $(W_{\max} - w)$ then $e'$ is matched in a maximum weight matching. Furthermore, since it is not essential, it is viable. The remaining edges are subpar. 

Next, for each vertex $v$, consider the game with player $v$ removed and find its worth. If it is less than $W_{max}$ then $v$ is essential. The endpoints of all viable edges are either essential or viable. Therefore, dropping the essential vertices from these yield viable vertices. Finally, a  vertex which does not have an essential or viable edge incident at it is subpar.   

Using the current best strongly polynomial weighted bipartite matching algorithm \cite{Kuhn1955Hungarian}, which takes $O(n^3)$ time, the classification requires a total time of $O(mn^3)$, see also Lemma \ref{lem.classify}.

\subsection{The Notion of Legitimate Tight Edges}
\label{sec.legitimate}

Next, we address the question: which tight edges are added to $T$ and when. By the rule given in Definition \ref{def.legitimate}, edge $(i, j)$ is declared legitimate when it is just about to go under-tight. Clearly, if the algorithm were to wait any longer, the edge would go under-tight and the dual would be infeasible. Hence the rule given supports the correctness of the algorithm.

\begin{example}
	\label{ex.exactly-tight}
In Figure \ref{fig.exactly-tight}, the solid edges are essential and dotted ones are subpar. Assume that the starting profit shares of the essential vertices $(u_1, v_1, u_2, v_2)$ are $(60, 40, 10, 50)$, respectively. Since $v_3$ is a subpar vertex it is moved to $\Fro$ and the two essential edges are moved to $\Init$. Also, $\Omega$ is initialized to $10$. Observe that subpar edge $(u_1, v_3)$ is tight at the start of the algorithm. The question is what is the ``right'' rule for adding a tight subpar edge $e$ to $T$? We study two rules:

{\bf Rule 1).}  {\em Subpar edge $e$ is added to $T$ exactly when it goes tight}.\\
If so, during Initialization, $(u_1, v_3)$ is added to $T$. At $\Omega = 10$, $(u_2, v_2)$ is moved from $\Init$ to $\Act$. $\Omega$ is raised at unit rate and when it becomes $20$, subpar edge $(u_1, v_2)$ becomes exactly tight and is added to $T$ and fundamental component $(u_1, v_1)$ is moved from $\Init$ to $\Act$. Now, since $(u_1, v_3)$ is in $T$, the entire graph moves into $\Fro$, and the algorithm ends with the imputation $(60, 40, 20, 40)$ for $(u_1, v_1, u_2, v_2)$. 

{\bf Rule 2).}  {\em Subpar edge $e$ is added to $T$ exactly when it is found to be legitimate}, see Definition \ref{def.legitimate}. \\
Under this rule, $(u_1, v_3)$ is not added to $T$ during Initialization. At $\Omega = 10$, $(u_2, v_2)$ is moved from $\Init$ to $\Act$. When $\Omega = 20$, subpar edge $(u_1, v_2)$ is found to be legitimate; it is added to $T$ and fundamental component $(u_1, v_1)$ is moved from $\Init$ to $\Act$. Now, the valid component consisting of edges $(u_1, v_1), (u_1, v_2), (u_2, v_2)$ undergoes repair. At $\Omega = 30$, it has min on both sides and is moved into $\Full$. Since $(\Init \cup \Act) = \emptyset$, the algorithm terminates with the imputation $(70, 30, 30, 30)$ for $(u_1, v_1, u_2, v_2)$. 

Rule 2 is the correct rule; its imputation leximin dominates that of Rule 1. 
\end{example}


\begin{figure}[h]
\begin{center}
\includegraphics[width=2.4in]{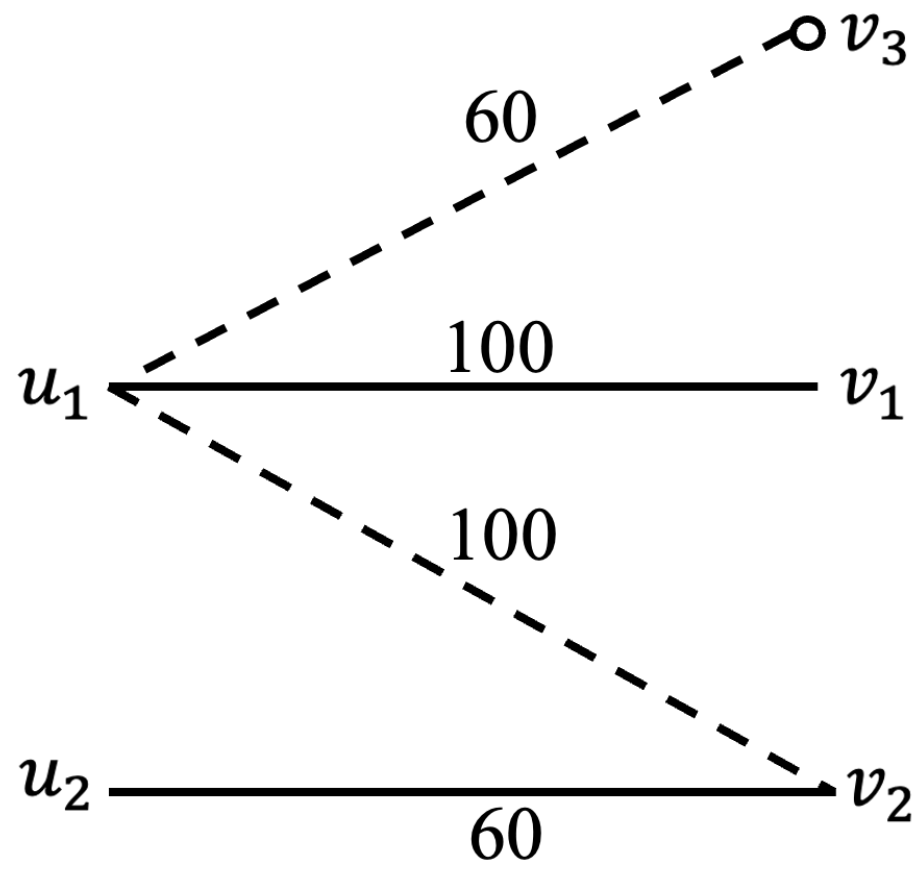}
\caption{The graph for Example \ref{ex.exactly-tight}.}
\label{fig.exactly-tight}
\end{center}
\end{figure}


By the rule given in Definition \ref{def.legitimate}, edge $(i, j)$ is declared legitimate when it is just about to go under-tight. Clearly, if the algorithm were to wait any longer, the edge would go under-tight and the dual would be infeasible. Hence the rule given supports the correctness of the algorithm.


\subsection{Primal Update: Dealing with Legitimate Subpar Edges} 
\label{sec.primal-update}

Let $e$ be a legitimate subpar tight edge which runs between components $C$ and $C'$, where $C \in \Act$. The following four cases arise; the cases can be executed in any order. Let $D$ be the component obtained in Step 2(c)(v) by merging the newly tight edge $(i, j)$ together with components $C$ and $C'$. For the exact definition of the minimum fully repaired sub-component of $D$ for each case, i.e., $\Ssuba(D), \Ssubb(D)$ and $\Ssubc(D)$, see Section \ref{sec.right}.
\begin{enumerate}
	\item $C' \in \Act$. If so, it must be the case that $C$ and $C'$ have min on opposite sides and the edge $(i, j)$ goes from one side of $C$ to the other side of $C'$. As a result,  $D$ has min on both sides. The algorithm moves the minimum fully repaired sub-component of $D$, $\Ssubb(D)$ to $\Full$ and the fundamental components of $(D - \Ssubb(D))$ to $\Init$.  
	\item $C' \in \Fro$. In order to prevent $(i, j)$ from going under-tight, we need to stop the repair of $C$. This is done by moving $\Ssubc(D)$ to $\Fro$ and the fundamental components of $(D - \Ssubc(D))$ to $\Init$. 
	\item $C' \in \Full$. If so, $\Ssubc(D)$ is moved to $\Full$ and the fundamental components of $(D - \Ssubc(D))$ are moved to $\Init$. 
	\item $C' \in \Init$. In order to keep repairing $C$, we will move $C'$ to $\Act$, and start repairing $D$. Observe that if $C$ and $C'$ have min on opposite sides and $\min(C') = \Omega$, then $D$ is fully repaired and before raising $\Omega$, Step 2(b) needs to be executed.  
\end{enumerate}

\begin{remark}
	\label{rem.legitimate}
Note that if several edges become legitimate simultaneously, they can be handled in arbitrary order. This is true even if the edges are incident at the same component, e.g., suppose two edges, $e_1$ and $e_2$, incident at component $C$ become simultaneously legitimate, with $e_1$ connecting into $\Fro$ and $e_2$ connecting into $\Full$. If so, $C$ will end up in $\Fro$ or $\Full$ depending on which edge is handled first. But in both cases no more changes will be made to the profit allocation of $C$ and therefore the core imputation computed will be the same. 
\end{remark}

At the moment when there are no more components left in $(\Init \cup \Act)$, $\Omega$ stops increasing, the current imputation is output (in Step 3) and the algorithm terminates. In Section \ref{sec.proof-min} we will prove that this imputation is indeed the leximin core imputation.

\subsection{Determining the ``Right'' Sub-Component}
\label{sec.right}

The question of determining the “right” sub-component of a valid component arises at four steps; however they can be classified into three different types of situations, as described below. Example \ref{ex.fully-repaired} illustrates the importance of determining the ``right'' sub-component for the first situation; it is easy to construct similar examples for the other situations as well. 

{\bf  Step 2(b):} 
A valid component $D$ has min on both sides.  Let $C_0$ be the root fundamental component of $D$, see Definition \ref{def.valid-comp}. First assume that $C_0$ has min in $\lleft(C_0)$. There is at least one descendent of $C_0$, say $C_i$, such that $C_i$ has min in $\rright(C_i)$ and $\min(C_0) = \min(C_i)$. Let $S$ be the set of all such descendents of $C_0$; clearly, $C_i \in S$. As described in Definition \ref{def.valid-comp} and proven in Lemma \ref{lem.tree-valid}, $D$ has a tree structure with root $C_0$. The {\bf minimum fully repaired} sub-component of $D$ is defined to be the smallest subtree of $D$ which contains $C_0$ and all components in $S$; we will denote it by $\Ssuba(D)$.  

Next assume that $C_0$ has min in $\rright(C_0)$. If so, in $S$ we will include all descendents, $C_i$ of $C_0$, such that $C_i$ has min in $\lleft(C_i)$ and $\min(C_0) = \min(C_i)$. The rest of the construction is same as above. Finally, if $C_0$ has min on both sides, $S$ will include both types of descendents defined above. 

The correct action to be taken in Step 2(b) is to move $\Ssuba(D)$ to $\Full$ and decompose $(D - \Ssuba(D))$ into fundamental components and move them back to $\Init$. See Example \ref{ex.fully-repaired} for a detailed explanation.


\begin{figure}[h]
\begin{center}
\includegraphics[width=2.2in]{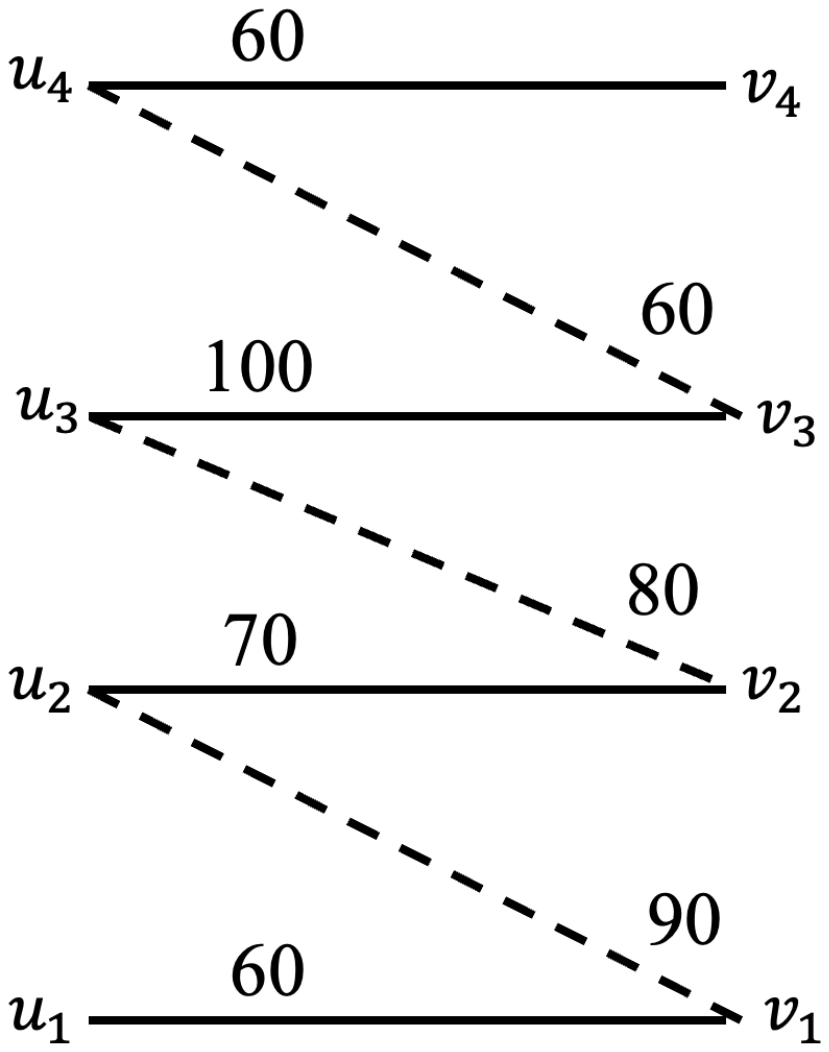}
\caption{The graph for Example \ref{ex.fully-repaired}.}
\label{fig.fully-repaired}
\end{center}
\end{figure}


\begin{example}
	\label{ex.fully-repaired}
Figure \ref{fig.fully-repaired} shows a valid component $D$ whose root component $C_0 = \{u_1, v_1\}$.  Assume that the profit shares of its essential vertices ($u_1, u_2, u_3, u_4, v_1, v_2, v_3, v_4$) are $(20, 50, 60, 20, 40, 20, 40, 40)$, respectively. Clearly, $D$ has min on both sides at $\Omega = 20$. Moving all of $D$ to $\Full$ will lead to a core imputation which is not leximin, as shown next.  

The minimum fully repaired sub-component of $D$, $\Ssuba(D)$, consists of the graph induced on ($u_1, u_2, v_1, v_2$), and it is added  to $\Full$ in Step 2(b) of the algorithm. Furthermore, the fundamental components of $(D - \Ssuba(D)$, namely ($u_3, v_3$) and ($u_4, v_4$) are moved them back to $\Init$. 

Next, while $\Omega = 20$, ($u_4, v_4$) is moved back to $\Act$ in Step 2(a). Also, since $(u_4, v_3)$ is a legitimate edge, in Step 2(c)(iii), ($u_3, v_3$) is also moved back to $\Act$. The two together form valid component $D' = (u_3, u_4, v_3, v_4)$. At this point, $\Omega$ starts increasing again. When $\Omega = 30$, $D'$ is found to have min on both sides. Furthermore, since $\Ssuba(D') = D'$, it is moved into $\Full$. 

The final core imputation computed for essential vertices ($u_1, u_2, u_3, u_4, v_1, v_2, v_3, v_4$) is given by $(20, 50, 70, 30, 40, 20, 30, 30)$, respectively. This is the leximin core imputation.
\end{example}

{\bf  Step 2(c)(iv):} 
Edge $(i, j)$ connecting valid components $D$ and $D'$, both in $\Act$, has gone tight; clearly, $(i, j)$ must be incident at opposite sides of $D$ and $D'$. Let $R$ be the component obtained by merging $D, D'$ and $(i, j)$ in Step 2(c)(ii). W.l.o.g. assume that $D$ has min in $\lleft(D)$ and $D'$ has min in $\rright(D')$. If so, $i \in \rright(D)$ and $j \in \lleft(D)$. Let $C_0$ and $C_0'$ be the root fundamental components of $D$ and $D'$, respectively and $C_1$ and $C_1'$ be the fundamental components of $D$ and $D'$ containing $i$ and $j$, respectively.  

Let $p$ and $p'$ be the paths from $C_0$ to $C_1$ and from $C_0'$ to $C_1'$, in $D$ and $D'$, respectively, as defined in Definition \ref{def.valid-comp}. Then the sub-component of $R$ consisting of $p$ and $p'$, together with the edge $(i, j)$, is the {\bf minimum fully repaired} sub-component of $R$; we will denote it by $\Ssubb(R)$. In Step 2(c)(iv), the algorithm moves $\Ssubb(R)$ to $\Full$ and decomposes $(R - \Ssubb(R))$ into fundamental components and moves them back to $\Init$.

{\bf  Steps 2(c)(v) and 2(c)(vi):} 
Edge $(i, j)$ connecting components $D$ and $D'$ has gone tight, where $D' \in \Act$, and $D \in \Fro$ in Step 2(c)(v) and $D \in \Full$ in Step 2(c)(vi); observe that $D$ need not be a valid component, as defined in Definition \ref{def.valid-comp}. Let $R$ be the component obtained by merging $D, D'$ and $(i, j)$ in Step 2(c)(ii). Let $C_0$ be the root fundamental component of $D'$  and $C_1$ be the fundamental component of $D'$ containing $i$.

Let $p$ be the path from $C_0$ to $C_1$ in $D'$, as defined in Definition \ref{def.valid-comp}. Then the sub-component of $R$ consisting of $p$ together with the edge $(i, j)$, is the {\bf minimum fully repaired} sub-component of $R$; we will denote it by $\Ssubc(R)$. In Step 2(c)(v), the algorithm moves $\Ssubc(R)$ to $\Fro$ and in Step 2(c)(vi), it moves $\Ssubc(R)$ to $\Full$. In either step, it decomposes $(R - \Ssubc(R))$ into fundamental components and moves them back to $\Init$.

{\bf  The Cleanup in Step 2(d):} 
The fundamental components which have been moved from $\Act$ back to $\Init$ may have subpar tight edges connecting them. All such edges are removed from $T$ in Step 2(d) so these components can be independently  improved in the future.

\begin{remark}
	\label{rem.Omega}
After the execution of Step 2(d) and before raising $\Omega$, it may turn out that for some fundamental component $C$ that just got moved from $\Act$ to $\Init$, $\min(C) = \Omega$. If so, $C$ needs to be moved back to $\Act$.  This may lead to the discovery of a legitimate edge incident at $C$ and the execution of a cascade of such events. All these events need to be executed before raising $\Omega$. 
\end{remark}

\begin{lemma}
	\label{lem.comp-added}
Let $D$ be the component that is moved to $(\Fro \cup \Full)$ at time $\Omega = t$. Then $\min(D) = t$.
\end{lemma}

\begin{proof}
	In all four ways of moving a component $D$ to $(\Fro \cup \Full)$, the root component of $D$, say $C_0$, is a also added. Since $\min(D) = \min(C_0) = t$, the lemma follows.  
\end{proof}

\begin{definition}
	\label{def.epoch}
We will partition a run of Algorithm \ref{alg.min} into {\bf epochs}. The first epoch starts with the start of  Algorithm \ref{alg.min} and an epoch ends when a component gets moved to $(\Fro \cup \Full)$, i.e., in Step 2(b), 2(c)(iv), 2(c)(v) or 2(c)(vi). This will happen when $\Omega = \alpha$, i.e., the max-min profit share, and one of the vertices in this component will have a profit of $\alpha$. Clearly there are at most $O(n)$ epochs. 
\end{definition}


\subsection{Extension to Degenerate Games}
\label{sec.char}

Recall that in Section \ref{sec.high}, we made the simplifying assumption that the given game is non-degenerate, i.e., has no viable edges or vertices. In this section we show how to deal with a general game. As before, let $H_0 = (U, V, T_0)$ be the subgraph of $G$ consisting of all essential and viable edges; by Theorem \ref{thm.edges}, these edges are tight under every core imputation. We will partition the connected components of $H_0$ into the two sets defined in Definitions \ref{def.unique-imp} and \ref{def.fundamental}.

 In the rest of this section, we will characterize these two types of components. Let us start by classifying the smallest components. A subpar vertex gets zero profit in all core imputations and is therefore a unique imputation component. By Lemma \ref{lem.unique}, an essential edge is a fundamental component; clearly, its  endpoints are essential vertices.

\begin{lemma}
	\label{lem.unique}
Each essential edge is a fundamental component.
\end{lemma} 

\begin{proof}
For the sake of contradiction, assume that essential edge $(i, j)$ has a unique imputation and let the profit shares of $i$ and $j$ be $a$ and $b$, respectively. Assume w.l.o.g. that $a \geq b$; clearly, $a > 0$. If $i$ does not have another edge incident at it in $G$, then changing the profit shares of $i$ and $j$ to $(0, a + b)$ will also lead to a core imputation.

Next assume that $i$ has $k$ other edges incident at it. These $k$ edges must necessarily be subpar and by Theorem \ref{thm.edges}, each of these edges, there is a core imputation under which the edge is over-tight, to the extent of at least $\epsilon > 0$, say. Now the convex combination of these core imputations is also a core imputation in which {\bf each} of these edges is over-tight to the extent of $\epsilon/k > 0$. Therefore, changing the profit shares of $i$ and $j$ to $(a - \epsilon/k, b + \epsilon/k)$ will lead to another core imputation. The contradiction establishes the lemma.  
\end{proof}

\begin{lemma}
	\label{lem.left.neq.right}
Let $C$ be a component of $H_0$ consisting of viable edges. If $|\lleft(C)| \neq |\rright(C)|$ then not all vertices of $C$ are essential and $C$ is a unique imputation component. 
\end{lemma}

\begin{proof}
Since $|\lleft(C)| \neq |\rright(C)|$, every maximum weight matching must leave an unmatched vertex in the bigger side. Therefore not all vertices of $C$ are essential.

Next, for contradiction assume that $C$ admits more than one core imputation. Let $v \in C$ get profit shares of $\alpha$ and $\alpha + \beta$, with $\beta > 0$, in two of these imputations, say $p_1$ and $p_2$. To keep all viable edges of $C$ tight, the profit shares of neighbors of $v$ under $p_2$ must be smaller than the corresponding profit shares under $p_1$ by $\beta$. Propagating these constraints alternately to $\lleft(C)$ and $\rright(C)$, we get that in going from $p_1$ to $p_2$, the profit shares of vertices in $\lleft(C)$ increase by $\beta$ and those in $\rright(C)$ decrease by $\beta$. Since $|\lleft(C)| \neq |\rright(C)|$, the total worth shared by the two imputations is different, leading to a contradiction.  
\end{proof}

\begin{example}
	\label{ex.neq}
	 The graph of Figure \ref{fig.viable-unbal} satisfies $|\lleft(C)| \neq |\rright(C)|$ and has the unique imputation of $(100, 0, 0)$ for the vertices $(u_1, v_1, v_2)$. Of these vertices, only $u_1$ is essential. 
\end{example}


\begin{figure}[h]
\begin{center}
\includegraphics[width=2.4in]{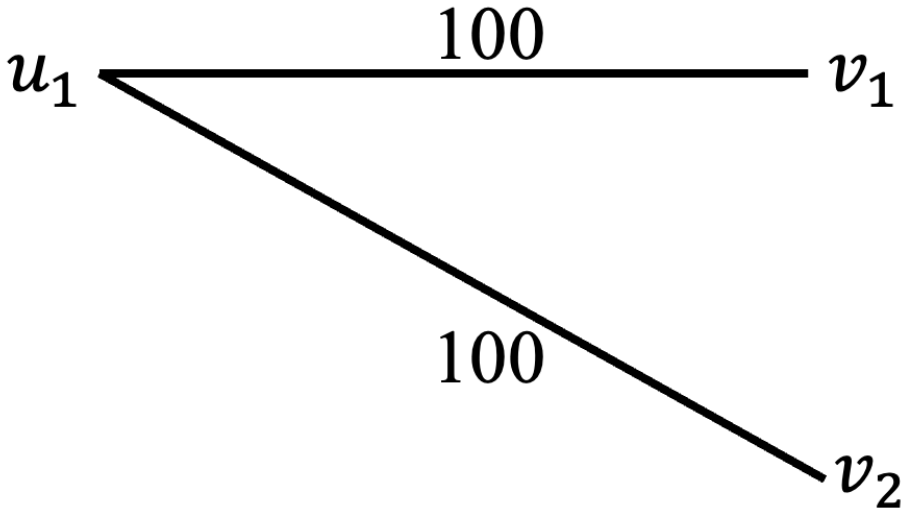}
\caption{A fundamental component $C$ satisfying $|\lleft(C)| \neq |\rright(C)|$ and therefore having a unique imputation.}
\label{fig.viable-unbal}
\end{center}
\end{figure}


\begin{lemma}
	\label{lem.unique-viable}
Let $C$ be a component of $H_0$ consisting of viable edges with $|\lleft(C)| = |\rright(C)|$. Then $C$ is fundamental if and only if all its vertices are essential.
\end{lemma}

\begin{proof}	
First assume that $C$ contains a viable vertex, say $u$. Clearly, all viable vertices of $C$ must get zero profit in all core imputations. We will show that the profit of each essential vertex is forced to be unique, thereby proving that $C$ is a unique imputation component. 

Consider an essential vertex $v$ of $C$. Let $p$ be a path from $u$ to $v$, consisting of vertices $u = v_0, v_1, v_2, \ldots , v_k = v$ and let the $k$ edges on $p$ have positive weights $w_1, w_2, \ldots , w_k$, respectively. We will argue that each vertex on $p$ must have a unique profit share; let us denote the profit share of $v_i$ by $f_i$. Since $v_0$ is viable, $f_0 = 0$. Since $(v_0, v_1)$ is tight, $f_1 = w_1 > 0$. Since edge $(v_1, v_2)$ is not allowed to be over-tight, $w_2 \geq f_1$. If $w_2 = f_1$, then $v_2$ must have a profit share of 0, and if $w_2 > f_1$, then $f_2 = w_2 - f_1 > 0$. In this manner, constraints can be propagated down $p$, eventually assigning a unique profit share to $v$. Hence $C$ is a unique imputation  component.

Next assume that all vertices of $C$ are essential and for the sake of contradiction assume that it has a unique core imputation. By Theorem \ref{thm.vertices}, each vertex of $C$ must get a positive profit. Moreover, as in the proof of Lemma \ref{lem.unique}, we can assume that all subpar edges incident at $C$ are over-tight by at least $\epsilon > 0$. Therefore, on increasing (decreasing) the profit shares of all vertices in $\lleft(C)$ ($\rright(C)$) by $\epsilon$, all edges of $C$ will remain tight and we get another core imputation, leading to a contradiction. Hence $C$ is a fundamental component.
\end{proof}

Theorem \ref{thm.classification} follows from Lemmas \ref{lem.unique}, \ref{lem.left.neq.right} and \ref{lem.unique-viable} and is used by the algorithm for partitioning the components of $H_0$ into fundamental and unique-imputation components. 

\begin{theorem}
		\label{thm.classification}
A component of $H_0$ is fundamental if and only if all its vertices are essential.
\end{theorem}

\begin{example}
	\label{ex.viable-bal}
Figure \ref{fig.viable-bal} shows a component $C$ satisfying $|\lleft(C)| = |\rright(C)|$ and having a unique core imputation; it is $(60, 0, 0, 40)$ for vertices $(u_1, u_2, v_1, v_2)$. Observe that $u_2$ and $v_1$ are not essential. Next, assume that $C$ is a complete bipartite graph with unit weight edges. Then $|\lleft(C)| = |\rright(C)|$ and all vertices of $C$ are essential. Clearly $C$ a fundamental component. 
\end{example}


\begin{figure}[h]
\begin{center}
\includegraphics[width=2.4in]{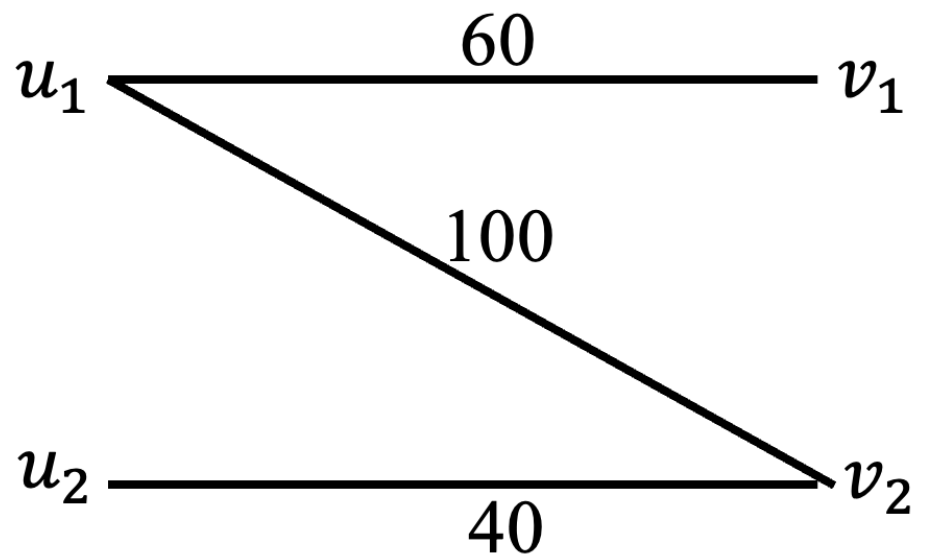}
\caption{The graph for Example \ref{ex.viable-bal}}.  
\label{fig.viable-bal}
\end{center}
\end{figure}



\section{Proof of Correctness of Algorithm \ref{alg.min}}
\label{sec.proof-min}

Our algorithm maintains the following second invariant:

{\bf Invariant 2:}   At any time $\Omega$ in the run of the algorithm, each fundamental component $C \in  (\Act \cup \Init)$ satisfies $\min(C) \geq \Omega$.

\begin{lemma}
	\label{lem.invariant} 
The algorithm maintains the Invariant 2. 
\end{lemma}

\begin{proof}
In Step 1(e), $\Omega$ is initialized to the minimum profit share of an essential vertex under the starting core imputation $(u, v)$, and at this time, each fundamental component $C$ with $\min(C) = \Omega$ is moved from $\Init$ to $\Act$. Therefore Invariant 2 is satisfied at the beginning of the algorithm. 

At each time $\Omega$, each fundamental component $C \in \Init$ satisfying $\min(C) = \Omega$ is moved from $\Init$ to $\Act$. Thereafter, the increase in $\Omega$ and $\min(C)$ is synchronized. Therefore Invariant 2 is maintained.

Assume valid component $D \in \Act$ and at time $\Omega$, $D$ has min on both sides, i.e., $D$ is fully repaired. Now, an appropriate sub-component of $D$ is moved from $\Act$ to $\Full$ in Step 2(c)(vi). The remaining fundamental components, say $C$, of $D$ are moved back to $\Init$. Since the latter satisfy $\min(C) \geq \Omega$, the fundamental components in $\Init$ still satisfy Invariant 2.

Another possibility is that at time $\Omega$ in Step 2(c), a subpar edge $e$ from valid component $D \in \Act$ to a component in $(\Act \cup \Fro \cup \Full)$ becomes exactly tight. If so, the appropriate step in Step 2(c) is executed. This may also involve moving some fundamental components, say $C$, back to $\Init$. Since the latter satisfy $\min(C) \geq \Omega$, the fundamental components in $\Init$ still satisfy Invariant 2.
\end{proof}

\begin{lemma}
	\label{lem.tree-valid}
At any time $\Omega$ in the run of the algorithm, each valid component $D \in \Act$ has the tree structure described in Definition \ref{def.valid-comp}. 
\end{lemma}

\begin{proof}
	The proof follows by an induction on the growth of $D$, i.e., the number of fundamental components in it. $D$ starts off when a fundamental component $C_0$ moves from $\Init$ to $\Act$. Assume $C_0$ has min in $\lleft(C_0)$, the other case is analogous. If so, $D$ will be rotated clockwise. Clearly, no subpar over-tight edge with both endpoints in $D$ can go tight due to this rotation. However, a subpar over-tight edge can go tight from $D$ to a fundamental component $C \in \Init$ or a valid component $D' \in \Act$. If so, the edge must go from $\rright(D)$ to $\lleft(C)$ or $\lleft(D')$. In the first case, $C$ is added to $D$ and the tree structure is maintained. In the second case, it must be the case that $D'$ is being rotated anti-clockwise, i.e., it has min in $\rright(D')$ and hence Step 2(c)(iv) kicks in and $D$ ceases to exist. The last possibility is that a subpar over-tight edge from $D$ to a component $D' \in (\Fro \cup \Full)$ goes tight. If so, Step 2(c)(v) or 2(c)(vi) kicks in and again $D$ ceases to exist.	
\end{proof}

Let $S$ be a set of unique imputation components and fundamental components of $G$ and let $W_S$ denote the set of vertices of all these components, $W_S \subseteq (U \cup V)$. Let $H(W_S)$ denote the subgraph of $G$ induced on $W_S$. The edges of $H(W_S)$ inherit weights from $G$, i.e., $w: E \rightarrow \cQ_+$. 

\begin{lemma}
	\label{lem.same}
	The following hold for the edge-weighted graph $H(W_S)$: 
\begin{enumerate}
	\item The projection of a maximum weight matching of $G$ onto $H(W_S)$ is a maximum weight matching in $H(W_S)$. Furthermore every maximum weight matching of $H(W_S)$ arises in this manner. 
	\item The projection, onto $H(W_S)$, of the partition of vertices and edges of $G$ into essential, viable and subpar is the correct partition for $H(W_S)$.   
\item The projection of a core imputation of $G$ onto $H(W_S)$ is a core imputation for $H(W_S)$. 
\end{enumerate}
\end{lemma}

\begin{proof}
\begin{enumerate}
	\item The crux of the matter lies in the following fact, which also gives the importance of moving entire  unique imputation and fundamental components of $G$ into the subgraph $H(W_S)$. Obtain $E'$ from $E$ by deleting all subpar edges and let $G' = (U, V, E')$. Clearly the sets of maximum weight matchings of $G$ and $G'$ are identical. Therefore each maximum weight matching in $G$ consists of a union of maximum weight matchings in the unique imputation and fundamental components of $G$. 
	
	Obviously the projection of a maximum weight matching of $G$ onto $H(W_S)$ is a matching in $H(W_S)$. Suppose for contradiction $H(W_S)$ has a heavier maximum weight matching, say $M$. Then $M$ together with maximum weight matchings in the fundamental components of $G$ which are not in $S$ will be a heavier matching in $G$, leading to a contradiction. 
 
	\item This is a corollary of the previous assertion. 
	
	\item Let $p$ be a core imputation of $G$ and let $w$ and $w'$ be the weights of maximum weight matchings in $G$ and $H(W_S)$, respectively. By definition, $p(W_S) \geq w'$. By the first part, the maximum weight matching in the induced subgraph on $((U \cup V) - W_S)$ has weight $w - w'$. Therefore $p((U \cup V) - W_S) \geq (w - w')$. Also, the worth of the game $G$ is 
	$$ w = p(W_S) + p((U \cup V) - W_S) \geq w' + (w - w') = w .$$
	Therefore $p(W_S) = w'$. 
	
	Finally observe that for any coalition $W' \subseteq W_S$, its worth is the same in $G$ and $H(W_S)$, and by definition, $p(W')$ is at least the worth of $W'$. Hence the projection of $p$ onto $H(W_S)$ is a core imputation for $H(W_S)$. 
\end{enumerate}
\end{proof}

At any time $\Omega = t$ in the run of the algorithm, consider the moment when $\Omega$ is just about to be increased, i.e, the cascade of events described in Remark \ref{rem.Omega} have all been executed. At this moment, let $W_t$ be the set of vertices of all components which are in $(\Fro \cup \Full)$; clearly $W_t \subseteq (U \cup V)$. Let $H_t$ denote the subgraph of $G$ induced on $W_t$.

\begin{lemma}
	\label{lem.H_t}
The profit shares of vertices of $W_t$ form the leximin core imputation for the graph $H_t$. 
\end{lemma}

\begin{proof}
By the third statement of Lemma \ref{lem.same}, the profit shares of vertices of $W_t$ form a core imputation for the graph $H_t$. Let $F$ denote the set of vertices of the unique imputation components in $\Fro$ at the start of the algorithm. The proof follows by an induction on the growth of $W_t$, i.e., the number of components (which in general are  unions of fundamental components and legitimate edges) that are added to $(\Fro \cup \Full)$. Clearly the assertion holds for $F$, which proves the basis of the induction.
	
For the induction step, assume that the assertion holds for $W_t$ and we wish to prove it for $W_{t'}$ where $t' > t$ and at time $t'$, one more component, say $D$, is added to $(\Fro \cup \Full)$ and hence to $W_{t'}$. Clearly, $\min(D) = t'$.  By Lemma \ref{lem.comp-added} each component $D'$ that was added to $(\Fro \cup \Full)$ until time $t$ satisfied $\min(D') \leq t$. 

There are two cases. If $D$ was added to $\Full$ because it was fully repaired, then its imputation cannot be leximin-improved and the assertion follows. If $D$ was added to $(\Fro \cup \Full)$ because a subpar edge $e$ just went tight, then the only way of leximin-improving $D$ is to leximin-deteriorate the current imputation of $W_t$ (since edge $e$ cannot be allowed to go under-tight). For each component $D'$ in $W_t$, $\min(D') <  \min(D)$. Therefore, this change will leximin-deteriorate the overall imputation for $W_{t'}$. Hence the current imputation of $W_{t'}$ is the leximin core imputation for the graph $H_{t'}$. The lemma follows. 
\end{proof}

\begin{theorem}
	\label{thm.progress}
Let $M$ denote the weight of a maximum weight edge in the given instance. Then, Algorithm \ref{alg.min} will terminate with a leximin core imputation by time $M/2$. 
\end{theorem}

\begin{proof}
By Lemmas \ref{lem.H_t}, and \ref{lem.invariant} at each time $t$, the profit shares of vertices of $W_t$ form the leximin core imputation for the graph $H_t$ and the profit share of each vertex in $(U \cup V) - W_t$ is $> t$.

Consider the extreme case that there is an essential edge $e$ having weight $M$ which starts off with profit shares of $0$ and $M$ on its endpoints. If $e$ does not get moved into $(\Fro \cup \Full)$ before time $M/2$, then it will be fully-repaired at time $M/2$. Clearly, all other fundamental components will get moved into $(\Fro \cup \Full)$ by that time either because they are fully-repaired or have a tight edge to a component in $(\Fro \cup \Full)$. The theorem follows.
\end{proof}

\begin{theorem}
	\label{thm.time}
	Algorithm \ref{alg.min} runs in strongly polynomial time; in particular, in time $O(mn^3)$. 
\end{theorem}

\begin{proof}
By Lemma \ref{lem.classify}, the time required to classify vertices and edges is $O(mn^3)$, see also Section \ref{sec.Classify}. Below we will show that the rest of the steps of Algorithm \ref{alg.min} require time $O(m n \log n)$, using the data structure of Fibonacci Heap \cite{Fibonacci.Heap}, which executes operations of insert and find-min in $O(1)$ time and delete-min and delete in $O(log n)$ time, amortized. Since the first term dominates the second, the theorem follows. 

There will be at most $m + n$ entries in the Fibonacci Heap, one corresponding to each edge and each fundamental component. The key of each entry specifies the time at which it will participate in an event. For each fundamental component $C \in \Init$, there is an entry whose key is $\min(C)$, since $C$ will move into $\Act$ when $\Omega = \min(C)$. Once a fundamental component $C$ moves into $\Act$, in Step 2(a) or 2(c)(iii), all subpar edges incident at $C$ are inserted in the Fibonacci Heap with the appropriate key. 

The minimum key in the Fibonacci Heap indicates the next event that will happen. Hence there is no need to increase $\Omega$ continuously; it can be increased in discrete steps. When Step 2(d) is executed, all the edges participating in the cleanup are deleted from the Fibonacci Heap. These $O(m)$ operations, which take $O(mn \log n)$ time, dominate the rest of the operations. 
\end{proof}

\bigskip


\begin{figure}

	\begin{wbox}
		\begin{alg}
		\label{alg.max}
		{\bf (Algorithm for Leximax Core Imputation)}\\

\begin{enumerate}
	\item {\bf Initialization:} 
	\begin{enumerate}
		\item $(u, v) \leftarrow$ an arbitrary core imputation. 		
		\item $T_0 \leftarrow $ essential and viable edges; $H_0 = (U, V, T_0)$. 
		\item $\Fro \leftarrow $ unique imputation components of $H_0$.
		\item $\Init \leftarrow $ fundamental components of $H_0$. 
		\item $\Omega \leftarrow$ maximum profit share of an essential vertex under $(u, v)$.  
		\item $T \leftarrow T_0$. 
	\end{enumerate}

\bigskip
		
	\item {\bf While} $(\Init \cup \Act) \neq \emptyset$  {\bf do}: \\
		   At unit rate, decrease $\Omega$ and repair all components $C \in \Act$.  \\
           \hspace*{2mm}  {\bf If}: 
		\begin{enumerate}
		\item $\exists C \in \Init \ s.t. \ \min(C) = \Omega$ {\bf then} move $C$ to $\Act$
		\item $\exists C \in \Act$ is fully repaired {\bf then} move $\Xsuba(C)$ to $\Full$ \\
		and the fundamental components of $(C - \Xsuba(C))$ to $\Init$. 
		\item $\exists$ legitimate edge $(i, j)$. \\
		Assume its endpoints are in $C$ and $C'$, where $C \in \Act$.  
		
			 \begin{enumerate} 
			 \item  $T \leftarrow T \cup \{(i, j)\}$. 
			 \item Merge $C, C'$ and $(i, j)$ into one component, say $D$. 
			 \item {\bf If} $C' \in \Init$ {\bf then} move $D$ to $\Act$.   
			 \item {\bf If} $C' \in \Act$ {\bf then} move $\Xsubb(D)$ to $\Full$\\ and the fundamental components of $(D - \Xsubb(D))$ to $\Init$.
			 \item {\bf If} $C' \in \Fro$ {\bf then} move $\Xsubc(D)$ to $\Fro$ \\ and the fundamental components of $(D - \Xsubc(D))$ to $\Init$. 
			 \item {\bf If} $C' \in \Full$ {\bf then} move $\Xsubc(D)$ to $\Full$ \\ and the fundamental components of $(D - \Xsubc(D))$ to $\Init$.
				\end{enumerate} 

		\item Cleanup $T$: remove all tight subpar edges\\  \hspace*{1.88cm} connecting fundamental components in $\Init$.

			\end{enumerate} 
		\item {\bf Output:} Output the current imputation and HALT.  
	\end{enumerate} 
			\end{alg}
	\end{wbox}
\end{figure}

\bigskip

\section{Algorithm for Leximax Core Imputation}
\label{sec.leximax}

Algorithm \ref{alg.max} computes the leximax core imputation for the assignment game, $G = (U, V, E), \ w: E \rightarrow \cQ_+$. It differs from Algorithm \ref{alg.min} in the following ways:

\begin{enumerate}
	\item In Step 1(e), $\Omega$ is initialized to the maximum profit share of an essential vertex under the initial core imputation $(u, v)$ computed in Step 1(a).   
	\item Whereas in Algorithm \ref{alg.min}, $\Omega$ is raised at unit rate, in Algorithm \ref{alg.max}, it is lowered at unit rate. 
	\item Definition \ref{def.max-both} plays a role analogous to that of Definition \ref{def.min-both}. Furthermore, for Algorithm \ref{alg.max}, we will say that component $C$ is {\bf fully repaired} if it has max on both sides. 
	\item Repairing a Component:	If component $C$ has max in $\lleft(C)$, then $C$ needs to be rotated anti-clockwise and if it has max in $\rright(C)$, it needs to be rotated clockwise. 
	\item The previous change will change the definition of valid component appropriately. Furthermore, the definitions of $\Ssuba, \Ssubb$ and $\Ssubc$ will change appropriately and we will denote the resulting sub-components by $\Xsuba, \Xsubb$ and $\Xsubc$, respectively. 
	\item The run of Algorithm \ref{alg.max} is partitioned into epochs via a definition analogous to Defintion \ref{def.epoch}. 
	\item The proofs of correctness and running time of Algorithm \ref{alg.max} are analogous to that of  Algorithm \ref{alg.min}. 
\end{enumerate}

\begin{definition}
\label{def.max-both}
	At any point in the Algorithm, for any connected component $C$ of $H$, $\max(C)$ will denote the {\bf largest profit share of an essential vertex in $C$}. We will say that {\bf $C$ has max on both sides} if there is a vertex in $\lleft(C)$ as well one in $\rright(C)$ whose profit is $\max(C)$. Otherwise, we will say that {\bf $C$ has max in $\lleft(C)$ $(\rright(C))$} if there is a vertex in $\lleft(C)$ ($\rright(C)$) whose profit is $\max(C)$. 
\end{definition}

\begin{theorem}
	\label{thm.max-poly-time}
Algorithm \ref{alg.max} computes the leximax core imputation in strongly polynomial time; in particular, in time  $O(mn^3)$. 

\end{theorem}

\section{Algorithm for Min-Spread Core Imputation}
\label{sec.min-spread}

\begin{definition}
	\label{def.alpha-beta}
Given the assignment game, $G = (U, V, E), \ w: E \rightarrow \cQ_+$, define $\alpha$ and $\beta$ to be the max-min and min-max profit of a vertex in a core imputation, respectively. 
\end{definition}

Our algorithm consists of the following two phases:

{\bf Phase 1:}  
Run Algorithm \ref{alg.min} on the assignment game, $G = (U, V, E), \ w: E \rightarrow \cQ_+$, until the end of the first epoch (see Definition \ref{def.epoch}). At this stage, the clock $\Omega = \alpha$. Let $(u, v)$ be the current core imputation. 

\bigskip 

{\bf Phase 2:}  
Run Algorithm \ref{alg.max} on $G = (U, V, E), \ w: E \rightarrow \cQ_+$, with the starting imputation being $(u, v)$ until:

{\bf Case 1:}
A component $C \in \Act$ attempts to decrease the profit of a vertex below $\alpha$. Output the current imputation and HALT.

{\bf Case 2:}
The end of the first epoch is reached; at this stage, $\Omega = \beta$. Output the current imputation and HALT.

\begin{theorem}
\label{thm.min-spread}
The output of this algorithm is a min-spread core imputation. 
\end{theorem}

\begin{proof}
By Lemma \ref{lem.invariant}, under the imputation $(u, v)$, each vertex has profit at least $\alpha$\footnote{Alternatively, in Phase 1, Algorithm \ref{alg.min} could be run all the way to the end to compute the leximin core imputation, which will also have this property.}.

{\bf Case 1:}
Component $C \in \Act$ must have a vertex, say $w$, having a profit of $\Omega$. Additionally, $C$ has a special vertex, say $w'$, whose profit is about to decrease below $\alpha$. These two vertices must be on the same side of $C$ and there must be a path of essential and tight subpar edges connecting them. Observe that decreasing the difference in profit of $w$ and $w'$, say $D$, will make one of these subpar edges go under-tight, leading to an infeasible dual. Therefore, $D$ is a lower bound on the spread of a core imputation. 

When the algorithm halts, each vertex has profit at most $\Omega$ and therefore the imputation output has spread $D$. Hence it is a min-spread core imputation.

{\bf Case 2:}
Clearly the spread of any core imputation must be at least $\beta - \alpha$. When the algorithm halts, each vertex has profit at most $\Omega = \beta$ and the profit of no vertex dropped below $\alpha$. Therefore the imputation output has spread $\beta - \alpha$ and hence it is a min-spread core imputation.  
\end{proof}

\begin{remark}
	\label{min-spread}
Another way of computing a min-spread core imputation is to run Algorithm \ref{alg.max} first and Algorithm \ref{alg.min} next—the details are analogous to the algorithm given above. In general it will lead to a different imputation. 
\end{remark}

\section{Discussion}
\label{sec.discussion}




Interestingly enough, the two principles underlying primal-dual algorithms also underlie algorithms for finding Nash equilibria for a two-player zero-sum game, namely: 

\begin{enumerate}
	\item  Finding globally optimal solutions by making only local improvements and 
	\item  Two processes, a minimizer and a maximizer, alternately making improvements to dual entities by responding to each other's progress.  
\end{enumerate}

For the Nash equilibrium problem, the two regret minimization algorithms take local improvement steps in the direction of the gradient, and the average of the strategies played converges to the set of Nash equilibria. This is a folklore fact that traces back to at least the works by Hannan \cite{Hannan1957approximation} and Foster and Vohra \cite{Foster1997calibrated}; see also the recent paper \cite{Ioannis}.

Finally, here are some open problems. 

\begin{enumerate}
	\item  Where in the polytope of optimal dual solutions to LP (\ref{eq.core-dual-bipartite}) do the leximin and leximax core imputations lie? Characterize these points; see also Section \ref{sec.framework}.
\item The running time given in Theorem \ref{thm.time} is dominated by the time required to classify vertices and edges, see Section \ref{sec.Classify}. Is there a way of speeding up this part thereby speeding up the entire  algorithm? 
\end{enumerate}

\section{Acknowledgements}
\label{sec.ack}

I wish to thank Felix Brandt, Naveen Garg, Kamal Jain and Herv{\'e} Moulin
 for valuable discussions, with special thanks to Rohith Gangam.

	\bibliographystyle{alpha}
	\bibliography{refs}

\end{document}